\documentclass[12pt,journal,draftclsnofoot,onecolumn]{IEEEtran}
\ifCLASSINFOpdf
\else
\fi

\usepackage{mathrsfs}
\usepackage{color}
\usepackage{amsfonts,amssymb}
\usepackage{amsfonts}
\usepackage{subfigure}
\usepackage{booktabs}

\usepackage{amsmath}
\usepackage{enumerate}
\usepackage{algorithm}
\usepackage{algorithmic}
\usepackage{bm}
\usepackage{makecell}

\usepackage{mathtools}
\usepackage{bbm}
\usepackage{dsfont}
\usepackage{pifont}
\usepackage{amsthm}
\usepackage{cite}
\usepackage{multirow}
\usepackage{tabularx}
\usepackage{threeparttable}

\newtheorem{remark}{Remark}
\newtheorem{theorem}{Theorem}
\newtheorem{lemma}{Lemma}
\newtheorem{corollary}{Corollary}
\newtheorem{assumption}{Assumption}

\newtheorem{definition}{Definition}

\let \sss=\scriptscriptstyle

\begin{document}

\title{
Local Topology Inference of Mobile Robotic Networks under Formation Control}
\author{Yushan Li$^{\dag}$, Jianping He$^{\dag}$, Lin Cai$^{\ddag}$ and Xinping Guan$^{\dag}$
	\thanks{
	${\dag}$: The Dept. of Automation, Shanghai Jiao Tong University, and Key Laboratory of System Control and Information Processing, Ministry of Education of China, Shanghai, China. E-mail address: \{yushan\_li, jphe, xpguan\}@sjtu.edu.cn. 

	${\ddag}$: The Dept. of Electrical and Computer Engineering, University of Victoria, BC, Canada. Email address: cai@ece.uvic.ca.

	Preliminary results have been presented at 2021 European Control Conference \cite{lys}.
	}%
}

\maketitle


\begin{abstract}
The interaction topology is critical for efficient cooperation of mobile robotic networks (MRNs). 
We focus on the local topology inference problem of MRNs under formation control, 
\textcolor{black}{where an inference robot with limited observation range can manoeuvre among the formation robots. }
This problem faces new challenges brought by the highly coupled influence of unobservable formation robots, inaccessible formation inputs, and unknown interaction range. 
\textcolor{black}{The novel idea here is to advocate a range-shrink strategy to perfectly avoid the influence of unobservable robots while filtering the input. 
To that end, we develop consecutive algorithms to determine a feasible constant robot subset from the changing robot set within the observation range, and estimate the formation input and the interaction range. 
Then, an ordinary least squares based local topology estimator is designed with the previously inferred information.} 
Resorting to the concentration measure, we prove the convergence rate and accuracy of the proposed estimator, taking the estimation errors of previous steps into account. 
\textcolor{black}{Extensions on nonidentical observation slots and more complicated scenarios are also analyzed. 
Comprehensive simulation tests and method comparisons corroborate the theoretical findings.}
\end{abstract}

\section{Introduction}
Mobile robotic networks (MRNs) have received increasing attention in the last decades. 
Thanks to the mobility, flexibility, and distributed fashion, MRNs are widely deployed, e.g., surveillance, reconnaissance, search and environmental monitoring \cite{bullo2009distributed}. 
Among these applications, formation control serves as a fundamental technique to enhance the cooperation performance by maintaining a preset geometric shape \cite{olfati2004consensus}. 
\textcolor{black}{
Numerous methods have been proposed to obtain stable and robust formation control, see \cite{oh2015survey,kamel2020formation} for a detailed review.} 
Despite the large variety of the control methods, the interaction topology among robots is universal and critical for effective cooperation of MRNs. 
The topology characterizes the locality of information exchange, and determines the shape-forming stability and convergence. 

\textcolor{black}{
Recent years have witnessed the emergence of many applications that necessitate advances in topology inference, which brings significant benefits in better understanding the system behaviors. 
Taking MRNs as the specific object, there are mainly two types of applications. 
First, from the security perspective, external attackers can utilize the topology inference method to find the critical robot that has significant control impacts in the formation, e.g., calculating the node degree and centrality \cite{zhang2014analysis}, or identifying the leadership relationship in the formation \cite{vasquez2018network}. 
With the topology information, more intelligent interception or herding tasks in military scenarios can be performed to control the formation \cite{choi2018detecting,li2019learning,licitra2019single,li2020unpredictable}. 
Second, from the perspective of performance improvement, inferring the topology of formation can support the self-configuration ability of MRNs \cite{venkitaraman2020recursive}. 
For instance, when a robot disconnects with others, it can use the inferred local topology to keep coordination with the formation, by predicting the state and reconnecting with appropriate neighboring robots \cite{li2019optimal}. }

Mathematically, topology inference can be seen as a typical inverse modeling problem. 
Plenty of related works have been developed for various dynamic models \cite{deka2016learning,shi2019bayesian,lu2019nonparametric,dong2019learning}. 
In relation to the basic consensus dynamics, the interaction topology is reconstructed by measuring the power spectral density of the network response to input noises, and node removal strategies are designed \cite{shahrampour2013reconstruction,shahrampour2014topology}. 
For sparsely connected dynamical networks, eigenvalue decomposition-based optimization methods in \cite{hassan2016topology,mateos2019connecting} are proposed to reconstruct the topology. 
\cite{bazanella2019network,van2021topology} investigate the identifiability conditions of the system topology of a class of heterogeneous dynamical networks, from the perspective of characterizing the system transfer matrix from input to output.
\textcolor{black}{
Despite the fruitful results, these methods cannot handle the topology inference of MRNs under formation control. 
For example, many well-established techniques are effective when the system is asymptotically stable and only involves zero-mean noises input \cite{matta2018consistent}, or the input is known \cite{coutino2020state,8985069}. 
Nevertheless, in practical formation control, the input is generally regular, the system can be marginally stable, and the state is not always fully observable by external observers. 
In a word, careful treatments of the formation input, interaction characteristics and observation limitations are still lacking. }

\textcolor{black}{
To fill the gap, this paper focuses on the local topology inference problem of MRNs under first-order linear formation control, 
where an inference robot can manoeuvre among the formation robots and observe their motions. 
Specifically, the inference robot has no knowledge of the formation inputs and interaction parameters, and the observation range is strictly limited. 
This problem is challenging due to three aspects. 
First, the set of robots within the observation range of the inference robot can change over time. 
Second, the movement of formation robots heavily depends on the unknown formation input and interaction constraints. 
Third, the state evolution of the observable robot subset is determined by not only itself but also the unobservable robots. 
It is quite difficult to decouple the influences of the mixed three factors, and obtain a reliable local topology from the noise-corrupted observations. 
To address these issues, the key insight is to determine an available robot set from the changing observable robot set, and eliminate the influence of the unobservable robots. 
Then, we need to filter the influence of the formation input from local observations and design an unbiased topology estimator. }

Preliminary results about estimator design with known interaction range have appeared in \cite{lys}. 
In this paper, we consider a more general situation where the interaction range is unknown, and extend the analysis by i) further estimating the unknown interaction range, ii) designing algorithms to determine the feasible robot set for inference, and iii) adding conjoint inference error analysis of the former two factors. 
The main contributions are summarized as follows. 

\begin{itemize}
\item We investigate the local topology inference problem of MRNs under noisy observations, {\color{black}{without the knowledge about the formation input and interaction parameters. 
By characterizing the steady formation pattern, we determine a constant subset from the time-varying set of robots within the observation range, and identify the formation input parameters. }}
The estimation error bound under finite observations is established in probability.

\item {\color{black}{
Leveraging the interaction constraints between formation robots, we develop an active excitation based method to obtain a reliable estimate of the interaction range. 
Combining the novel range-shrink strategy and the monotonicity analysis of the interaction range, the influence of unobservable robots is perfectly avoided. 
Then, an ordinary least squares (OLS) based local topology estimator is established after filtering the formation input's influence on observations before the steady stage. }}

\item 
The convergence and accuracy of the proposed estimator are proved, by resorting to the concentration measure with probability guarantees. 
\textcolor{black}{
Extensions on nonidentical observation slots of the robots and on more complicated control models are also discussed and analyzed. 
Simulation studies and comparison tests illustrate the effectiveness of the proposed method.  }
\end{itemize}

\textcolor{black}{
This paper reveals the possibility of inferring the local topology of MRNs under first-order linear formation control protocols, without knowledge about the formation input and interaction parameters. 
The achieved results provide insights to tackle more complicated and general scenarios, and also necessitate the investigation of interaction security of MRNs. 
}

The remainder of this paper is organized as follows. 
Section \ref{r-work} presents related literature. 
Section \ref{preliminary} gives the modeling for MRNs and formulates the inference problem. 
Section \ref{revealing} studies how to identify the steady pattern and interaction range. 
Section \ref{sec:inference-estimation} develops the design of the local topology estimator and analyzes the inference performance. 
Simulation results are shown in Section \ref{simulation}, 
followed by the concluding remarks and further research issues in Section \ref{conclusion}. 
All the proofs of theorems are provided in the Appendix.

\section{Related Work}\label{r-work}
\textit{Formation control in MRNs}. 
The fundamental rules for formation control were first introduced by the famous Reynolds' Rules \cite{reynolds1987flocks}: separation, alignment, and cohesion. 
Based on the rules, numerous methods have been proposed to achieve the desired performance, and consensus-based algorithms have become the mainstream, e.g., \cite{sun2016optimal,zhao2018affine,alonso2019distributed,xu2020affine}. 
The key idea of consensus-based algorithms is that the formation is modeled as a graph, and every robot exchanges information (positions and velocities) with its neighbors and computes its control inputs. 
Therefore, the interaction structure lays critical support for effective formation control and \textcolor{black}{is largely affected by communication network}. 
In recent years, communication-free formation control \cite{deghat2014localization,cheng2017event,trinh2018bearing} has been developed and attracts research interests, thanks to the fast advancement of sensing technologies. 
Communication-free interaction avoids information delays and network bandwidth consumption, and even enables stealth modes of operation \cite{kan2011network}. 
For instance, formation control with bearing measurements by vision sensors was investigated in \cite{zhao2019bearing}. 
\textcolor{black}{Note that in all cases, the interaction range is restricted by the physical distance between robots} due to the energy constraints, i.e., two distant robots outside the interaction range are disconnected. 

\textit{Topology Inference}. 
A large body of research concerning topology inference has been developed in the literature. 
\cite{granger1969investigating,brovelli2004beta} used Granger causality to formulate the directionality of the information exchange among system nodes, and constructed corresponding estimators to infer the underlying topology. 
\textcolor{black}{
Identifying the topology of sparsely connected networks via compressed sensing is also commonly investigated \cite{timme2007revealing,wang2011network,hayden2016sparse,wai2019joint}, which is transformed to a constrained $L_1$ norm optimization problem based on limited observations. 
Considering the latent regularity in the time series of nodal observations and adopting some basic assumptions (e.g., smoothness), graph signal processing methods \cite{mei2015signal,onuki2016graph,egilmez2017graph,hallac2017network,pasdeloup2018characterization} are proposed to derive a topology interpretation for the causation or correlation between nodes. 
}
When the network dynamics are nonlinear, kernel-based methods were developed to effectively infer the topology \cite{karanikolas2016multi,karanikolas2017multi,wang2018inferring}. 
The key idea is to select appropriate kernel functions to approximate the nonlinearities, where the performance is mainly determined by the kernel design. 
\textcolor{black}{Several works \cite{vasquez2018network,8985069} have directly considered inferring the topology of MRNs, but they still lack performance guarantees, especially when the knowledge about the formation input is unavailable. }

In summary, most existing works cannot directly infer the topology of MRNs under formation control, due to the unknown formation input and interaction characteristics. 
Despite many attempts on the asymptotic inference performance, there is no analytical model for the inference error under finite observations. 
These challenges motivate this paper.

\section{Preliminaries and Problem Formulation}\label{preliminary}
Let $\mathcal{G}=(\mathcal{V},\mathcal{E})$ be a directed graph that models an MRN, where $\mathcal{V}=\{1, \cdots,n\}$ is a finite set of nodes (i.e., robots) and $\mathcal{E}\subseteq \mathcal{V}\times \mathcal{V}$ is the set of interaction edges. 
An edge $(i,j)\in \mathcal{E}$ indicates that $i$ will use the information from $j$. 
The adjacency matrix $A=[a_{ij}]_{n \times n}$ of $\mathcal{G}$ is defined such that ${a}_{ij}\!>\!0$ if $(i,j)$ exists, and ${a}_{ij}\!=\!0$ otherwise. 
Denote ${\mathcal{N}_i^{in}}=\{j\in \mathcal{V}:a_{ij}>0\}$ and ${\mathcal{N}_i^{out}}=\{j\in \mathcal{V}:a_{ji}>0\}$ as the in-neighbor and  out-neighbor sets of $i$, respectively. 

Throughout the paper, we use the scripts $\tilde \cdot $ and  $\hat \cdot$ right above a variable to indicate the corresponding observation and estimator, respectively. 
We denote by $\| \cdot \|$ the spectral norm and by $\| \cdot \|_{F}$ the Frobenius norm of a matrix. 
Denote $\bm{0}$ by all-zero matrix and $\bm{1}$ by all-one matrix in compatible dimensions. 
The set variables are expressed in capital calligraphy fonts, and $\mathcal{V}_a \backslash \mathcal{V}_b$ represents the elements in $\mathcal{V}_a$ that are not in $\mathcal{V}_b$. 
The two-dimension state of a robot is expressed in boldface font (e.g., ${\mathbf z}$). 
Unless otherwise noted, the formulation with non-boldface state variables applies to the robot state in each dimension independently. 
For square matrices $M_a$ and $M_b$ in the same dimensions, ${M_a}\!\succeq\!{M_b}$ (${M_a}\!\preceq\!{M_b}$) means ${M_a}-{M_b}$ is positive-semidefinite (negative-semidefinite). 
For two real-valued functions $f_1$ and $f_2$, $f_1(x)=\bm{O}(f_2(x))$ as $x\to x_0$ means $\mathop {\lim }\nolimits_{x \to x_0 } |f_1(x)/f_2(x)|<\infty$, and $f_1(x)=\bm{o}(f_2(x))$ as $x\to x_0$ means $\mathop {\lim }\nolimits_{x \to x_0 } |f_1(x)/f_2(x)|=0$. 
\textcolor{black}{Some important symbols are summarized in Table \ref{tab:test}. }

\begin{table}[t]
\small
\centering
 \caption{\label{tab:test}Some Important Notation Definitions} 
 \begin{tabular}{cl}
 \toprule 
  Symbol  & Definition  \\ 
  \midrule
  $r_a$, $r_i$ & the abbreviation of the inference robot, robot $i$\\
  $z^a_k$, $z^i_k$ & the state of $r_a$, $r_i$ at time $k$\\
  ${\mathbf z}_k^{a}$, ${\mathbf z}_k^{i}$ & the two-dimensional position of $r_a$, $r_i$ at time $k$ \\
  $c$ & the desired velocity of formation robots\\
  $h$ & the shape configuration vector of formation robots\\
  $k_s$ & the time when $\epsilon$-steady pattern is reached\\
  $k_{end}$ & the time when $r_a$ stops observation\\
  $\mathcal{V}_{\sss F}^{a}(k)$ & the robot set within $r_a$'s observation range at time $k$ \\  
  $\mathcal{V}_{\sss F}$ & the constant robot subset observed by $r_a$\\  
  $\mathcal{V}_{\sss H}$ & the robot subset by range-shrink strategy ($\mathcal{V}_{\sss H}\subseteq\mathcal{V}_{\sss F}$)\\ 
  $z^{\sss F}_k$, $z^{\sss H}_k$ & the state vector of robot set $\mathcal{V}_{\sss F}$, $\mathcal{V}_{\sss F}$ at time $k$\\
  $W$ & the interaction topology matrix among the formation\\
  $W_{\sss HF}$ & the interaction topology matrix between $\mathcal{V}_{\sss H}$ and $\mathcal{V}_{\sss F}$\\
  $R_f$ & the observation range of $r_a$\\
  $R_c$ & the interaction range of formation robots\\
  $R_o$ & the obstacle detection radius of formation robots\\
  $X$ & the matrix of $k_s$ filtered observations about $\mathcal{V}_{\sss F}$\\ 
  $Y$ & the matrix of $k_s$ filtered observations about $\mathcal{V}_{\sss H}$\\ 
  \bottomrule 
 \end{tabular} 
\end{table}

\subsection{Formation Control} \label{s2-c}
To describe the predefined geometric shape under formation control, 
\textcolor{black}{the shape vector $h_0=[h_0^1,\cdots,h_0^n]^{\mathsf{T}}$ is introduced, where $h_0^i(i\in\mathcal{V})$ is the desired relative deviation between robot $i$ (abbreviated to $r_i$ hereafter) and a common reference point.} 
To achieve this pattern, a common first-order discrete consensus-based controller is given by \cite{olfati2007consensus}
\begin{equation}\label{eq-1}
z^{i}(t_{k+1})=z(t_k)+\varepsilon_{\sss T} \sum\limits_{j \in \mathcal{N}_i^{in}} {{a_{ij}}(z^j(t_k)-z^i(t_k)-h_0^{ij})},
\end{equation}
where $h_0^{ij}=h_0^j-h_0^i$ is desired state deviation between $j$ and $i$, and $\varepsilon_{\sss T}=t_{k+1}-t_{k}$ is the control period satisfying $\varepsilon_{\sss T}\le 1/\max\{d_i: i\in\mathcal{V}\}$. 
Note that once the formation shape is specified, the choice of the reference point will make no difference as $h_0^{ij}$ remains unchanged. 

Generally, to dynamically guide the formation motion, one robot will be specified as the leader with an extra velocity input. 
For simplicity and without loss of generality, $r_{n}$ is taken as the leader and reference node, and suppose that it runs in a constant velocity $c_0$. 
\textcolor{black}{Let $L=\text{diag}\{A \bm{1}_n\}-A$ be the Laplacian matrix of $\mathcal{G}$}, and denote $u_0=Lh_0+[0,\cdots,0, c_0]^\mathsf{T}$. 
Then, the global dynamics of the system is described by 
\begin{equation}\label{eq:original-global}
z(t_{k+1})= ( {I}_{n}-\varepsilon_{\sss T} L) z(t_k)+ \varepsilon_{\sss T}u_0 {\buildrel \Delta \over =} W z(t_k)+ \varepsilon_{\sss T}u_0,  
\end{equation}
where $W$ equivalently represents the original topology matrix $A$ and is known as Perron matrix. 
Apparently, $W$ is row-stochastic, i.e., $W\bm{1}_{n}=\bm{1}_{n}$. 
For ease of notation, we denote $z_{k} {\buildrel \Delta \over =} z(t_k)$, $c{\buildrel \Delta \over =}\varepsilon_{\sss T}c_0$, $h{\buildrel \Delta \over =}\varepsilon_{\sss T}h_0$ and $u{\buildrel \Delta \over =}\varepsilon_{\sss T}u_0$ in following sections. 
Then, (\ref{eq:original-global}) is rewritten as 
\begin{equation}\label{eq:global-system}
z_{k+1}= W{z_k}+ u. 
\end{equation}
We make the following assumption throughout this paper. 

\begin{assumption}[System stability]\label{assum:stability}
The eigenvalue 1 of $W$ is simple \textcolor{black}{(i.e., its algebraic multiplicity equals one)}, and the magnitudes of all other eigenvalues are less than one. 
\end{assumption}


\subsection{Obstacle-avoidance and Interaction Constraints}
The obstacle-avoidance mechanism is critical for MRNs to interact with the physical environment. 
Denote by $R_o$ the the obstacle detection range, and by \textcolor{black}{$u_{k}^{j,e}$ the input triggered by the excitation source (i.e., the obstacle $r_{ob}$) on $r_j$. 
Once the relative distance between $r_j$ and $r_{ob}$ satisfies $\|\mathbf{z}^j-\mathbf{z}^{ob}\|\le R_o$, the state of $r_j$ is updated by 
\begin{align}\label{eq:obstacle-rule}
z_{k+1}^{j,e}=\sum\limits_{\ell \in \mathcal{V}} {{w_{j\ell}}(z_k^{\ell}-z_k^j)} + u_k^{j} + u_{k}^{j,e},
\end{align}
where the first two terms on the right hand side (RHS) can be seen as the internal interaction within the MRN}, while the last term represents the external interaction with the environment. 

There are numerous obstacle-avoidance algorithms in the literature (e.g., \cite{pandey2017mobile} provides a detailed review), and among them, $u_{k}^{j,e}$ is mainly determined by the desired goal state, the relative state and velocity between $r_{j}$ and $r_{ob}$. 
As long as the excitation source appears within the obstacle-detection range of $r_j$, there will always be a $u_k^{j,e}\neq0$. 
In this work, we do not specify the detailed form of $u_{k}^{j,e}$, but mainly leverage the obstacle-avoidance property that  
\begin{equation}\label{eq:obstacle-property}
|u_{k}^{j,e}|>0,~\text{if}~\|\mathbf{z}^j-\mathbf{z}^{ob}\|\le R_o. 
\end{equation}
In practical applications, the interaction capability of robots is limited due to the energy constraint, and thus the interaction range among robots (denoted by $R_c$) is bounded \cite{bullo2009distributed}, satisfying 
\begin{equation}\label{eq:interaction-constraint}
R_o<R_c<\infty.
\end{equation}

\begin{figure*}[t]
\centering
\setlength{\abovecaptionskip}{0.1cm}
\includegraphics[width=0.8\textwidth]{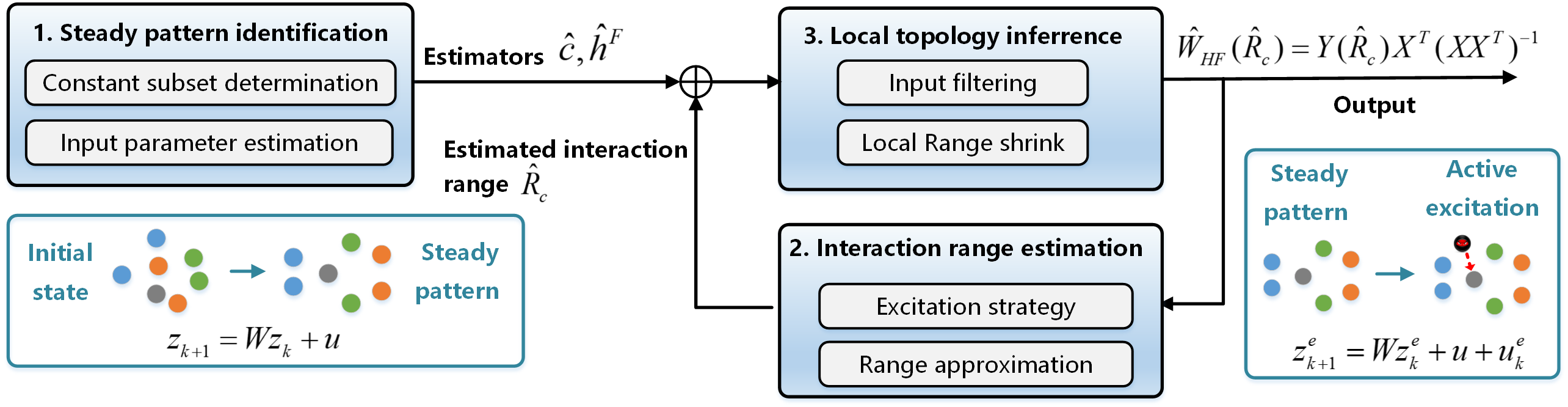}
\caption{The proposed local topology inference method. 
First, the inference robot $r_a$ uses the collected observations over the MRN to estimate the formation input parameters $c$ and $h^{\sss F}$. 
Then, $r_a$ makes active excitations on a target robot in the MRN to estimate the interaction range between two robots. 
Finally, based on the estimated information, $r_a$ can filter the influence of the unobservable part and determine the shrunken range. 
Specifically, the inferred topology can be leveraged in turn to approximate the best shrink and infer a new local topology. 
}
\label{fig:frame}
\end{figure*}

\subsection{Problem of Interest}

\textcolor{black}{
Suppose an inference robot (denoted by $r_a$) can manoeuvre in an MRN described by the formation control model \eqref{eq:global-system}. 
Specifically, $r_a$ is equipped with advanced sensors with a limited observation range, and does not have knowledge about the formation input and interaction parameters. 
Note that both the formation robots and $r_a$ are moving during the whole process, and thus the robots within the observation range of $r_a$ can change over time. 
Let $\mathcal{V}_{\sss F}^{a}(k)\subseteq\mathcal{V}$ be the set of robots within $r_a$'s observation range at time $k$, given by 
\begin{equation}
\mathcal{V}_{\sss F}^{a}(k)=\{ i: \|{\mathbf z}_{k}^{i}-{\mathbf z}_{k}^{a}\|_2 < R_{f} \},
\end{equation} 
where $R_f$ is the observation range of $r_a$. 
Since there can be possible observation inaccuracies brought by the movement of robots, $r_a$'s observation for $i\in\mathcal{V}_{\sss F}^{a}(k)$ is described by 
\begin{equation}
\tilde z_{k}^i=z_{k}^i+\omega_{k}^i,~i\in\mathcal{V}_{\sss F}^{a}(k),
\end{equation}
where $\omega_{k}^i$ is the $i$-th element of \textit{i.i.d.} Gaussian noise vector $\omega_{k}\in\mathbb{R}^{n}$, satisfying $\omega_{k}\sim N(0,{\sigma ^2}I)$. 
Considering the interaction constraint (\ref{eq:interaction-constraint}) in $\mathcal{V}$, we assume that $R_f$ satisfies
\begin{equation}\label{eq:observation-constraint}
R_f\ge R_c,
\end{equation}
where implicates that $r_a$ can observe at least one single robot and all its in-neighbors. }

{\color{black}{
The goal of this paper is to investigate how $r_a$ can infer the local topology of the formation from the observations $\{ \tilde z_{k}^i,~i\in\mathcal{V}_{\sss F}^{a}(k)\}$. 
This problem is very challenging, and most existing methods cannot be directly applied due to three factors: 
i) Time-varying $\mathcal{V}_{\sss F}^{a}(k)$: the observations of robots in $\mathcal{V}_{\sss F}^{a}(k)$ may be discontinuous and insufficient. 
ii) Weak prior knowledge: the unknown formation input and interaction parameters make direct inference from $\{ \tilde z_{k}^i,i\in\mathcal{V}_{\sss F}^{a}(k)\}$ unavailable. 
iii) Limited observation range: the neighbors that send real-time information to $\mathcal{V}_{\sss F}^{a}(k)$ may locate outside the observation range of $r_a$. 
We will address these issues from the following aspects to obtain a reliable local topology inference. 
\begin{itemize}
\item Utilizing the steady pattern of the formation, we first demonstrate how to determine a constant subset $\mathcal{V}_{\sss F}\subseteq\mathcal{V}_{\sss F}^{a}(k)$ as available inference sources, and identify the formation input from corresponding observations. 

\item Since the interaction range between robots is limited, we develop an excitation method to estimate the interaction range, and later use it to improve the local topology inference performance. 

\item Towards the influence of unobservable robots on $\mathcal{V}_{\sss F}$, we propose a novel range-shrink method to guarantee the inferred topology is unbiased in the asymptotic sense. 
\end{itemize}
Based on the above treatments, we finally present the local topology estimator, along with its convergence and accuracy analysis. 
Specifically, the situation that the observation slots for robots in $\mathcal{V}_{\sss F}$ are nonidentical will also be analyzed. 
The whole framework of this paper is shown in Fig.~\ref{fig:frame}. 
}}

\section{Estimating the Steady Pattern and the Interaction Range}\label{revealing}
In this section, 
\textcolor{black}{we first demonstrate how to determine a constant subset $\mathcal{V}_{\sss F}$ from $\mathcal{V}_{\sss F}^{a}(k)$ and identify the formation input. }
Then, we present the range-shrink idea by introducing a common truncated estimator. 
\textcolor{black}{Finally, the excitation strategy for estimating the interaction range is provided}

{\color{black}{
\subsection{Determining Constant Robot Subset $\mathcal{V}_{\sss F}$}

Suppose the MRN starts the formation task from an arbitrary initial state. 
Given the initial position of $r_a$, $r_a$ needs to manoeuvre among the formation robots and avoid collisions with them, namely, keeping $\|{\mathbf z}_{k}^{a}\!-\!{\mathbf z}_{k}^{i}\|_2 \!>\!R_o, i\in\mathcal{V}_{\sss F}^a(k)$. 
This can be easily achieved by making $r_a$ not too close to the robots and track the formation velocity, e.g., setting 
\begin{equation}\label{eq:movement}
u^a_k=\sum\nolimits_{i\in \mathcal{V}_{\sss F}^{a}(k) } ( z_k^i-z_{k-1}^i )/|\mathcal{V}_{\sss F}^{a}(k)| + g_a(\mathcal{V}_{\sss F}^{a}(k)),
\end{equation}
where the first sum term is for formation tracking, and $g_a(\mathcal{V}_{\sss F}^{a}(k))$ represents the adjusting input when $r_a$ is too close to some robots. 
Note that any strategy that meets the above requirement can be adopted by $r_a$. 
Then, we focus on how to infer the local topology from $r_a$'s observations in this process. 
}}

Since the steady pattern of the MRN reflects the formation shape and moving speed of the MRN, 
we first characterize the steady pattern by introducing the notion of linear steady trajectory, and determine the subset $\mathcal{V}_{\sss F}$ to be inferred. 
\begin{definition}[Linear steady trajectory]\label{def:steady-pattern}
Given the dynamic system (\ref{eq:global-system}), its state evolution $\{z_{k}\}$ is subject to linear steady trajectory if there exists unique $c\in\mathbb{R}$ and $s\in\mathbb{R}^{n}$ such that 
\begin{equation}
z_{k}=ck\bm{1}_n + s. 
\end{equation}
\end{definition}
\textcolor{black}{By referring to the Theorem 1 in our preliminary work \cite{lys}, we have the following result about the steady trajectory. }
\begin{lemma}\label{le:ep-convergence}
By the constant controller $u=Lh+[0\cdots0 \;c]^\mathsf{T}$, the system (\ref{eq:global-system}) will approximate the linear steady trajectory with arbitrary precision, i.e., 
given an arbitrary $\epsilon>0$, there always exists a $k_0\in\mathbb{N}^{+}$ and a unique $s\in\mathbb{R}^{n}$, such that 
\begin{equation}\label{eq:steady-state}
\|z_k-{c} k \bm{1}_n -s \|_1<\epsilon, \forall k\ge{k_0}. 
\end{equation}
\end{lemma}

{\color{black}{
Lemma \ref{le:ep-convergence} illustrates that when the formation is in the linear steady trajectory with tolerant accuracy $\epsilon$ (we call it as $\epsilon$-steady pattern hereafter), all robots are running at a common speed with fixed relative state deviations. 
Utilizing this property and given appropriate following strategy for $r_a$, 
we have the following result. 
\begin{lemma}\label{le:steady-subset}
Given an arbitrary $\epsilon>0$, there always exists a $k_1\in\mathbb{N}^{+}$, $\forall k\ge k_1$, $\mathcal{V}_{\sss F}^{a}(k)$ remains unchanged. 
\end{lemma}
Lemma \ref{le:steady-subset} follows easily from Lemma \ref{le:ep-convergence}. 
Taking the moving strategy \eqref{eq:movement} as an example, 
when the formation reaches $\epsilon$-steady pattern, $r_a$ will also move stably with the MRN with almost the same velocity, and thus the formation robots in the observation range of $r_a$ will not change. 
Based on the analysis, we determine the constant local subset $\mathcal{V}_{\sss F}$ by 
\begin{equation}
\mathcal{V}_{\sss F}=\mathcal{V}_{\sss F}^{a}(k_{end}),
\end{equation}
where $k_{end}$ represents the time when $r_a$ stops observing the MRN. 
For simplicity, we temporarily assume $\mathcal{V}_{\sss F}\subseteq\mathcal{V}_{\sss F}^{a}(k)$ for an arbitrary $k$, and extend the analysis to the cases when this assumption is violated in Section \ref{subsec:extension}.  
}}

\subsection{Steady Pattern Identification}\label{subsec:steady}

{\color{black}{
After the local set $\mathcal{V}_{\sss F}$ is determined, the steady pattern parameters of the formation can be identified from the observations by utilizing Lemma \ref{le:ep-convergence}. 
}}
Based on (\ref{eq:steady-state}) and taking the observation noises into account, if the formation has reached $\epsilon$-steady pattern, then the pattern parameters can be identified by solving
\begin{equation} \label{eq:solving-steady}
\mathop {\min }\limits_{c,s^{\sss F}} \sum\limits_{t = {k}}^{k + L_c} {{{\left\| {\tilde z_t^{\sss F}} - ct \bm{1}_{n_f}  +  s^{\sss F} \right\|}_2^2}},
\end{equation}
\textcolor{black}{where ${\tilde z_t^{\sss F}}=[\tilde z_t^i,i\!\in\!\mathcal{V}_{\sss F}]\!\in\!\mathbb{R}^{n_f}$ represents the observation vector of $\mathcal{V}_{\sss F}$ at time $t$, }
$n_f=|\mathcal{V}_{\sss F}|$, and $L_c$ is the observation window length. 
Note that (\ref{eq:solving-steady}) is a typical least squares problem, whose solution is given by 
\begin{itemize}
\item \textbf{Steady pattern estimator}:
\begin{equation} \label{eq:window-s}
\!\!\left \{
\begin{aligned}
\hat c(k,L_c)  &\!=\!{\sum\nolimits_{t= k}^{k+L_c-1} \bm{1}_{n_f}^\mathsf{T} (\tilde{z}_{t+1}^{\sss F}-\tilde{z}_{t}^{\sss F}) } /{ ({n_f}{L_c}) },  \\
\hat{s}^{\sss F}(k,L_c) &\!=\! {\sum\nolimits_{t = k+1}^{k+L_c} (\tilde{z}_{t}^{\sss F} -\hat{c} t \bm{1}_{n_f}) } / {L_c}.
\end{aligned}\right.
\end{equation}
\end{itemize}

Next, we demonstrate the estimation performance of (\ref{eq:window-s}). 
\begin{theorem}[Accuracy of $\hat c$ and $\hat{s}^{\sss F}$]\label{th:cs-performance}
Suppose the MRN has reached $\epsilon$-steady pattern after $k_0$. 
Let $\Delta_c= \hat c(k_0,L_c) - c $ be the estimation error of $\hat c$, then we have 
\begin{equation} \label{eq:accuracy-c}
\Pr\left\{ | \Delta_c | \le \frac{4\epsilon}{\sqrt{L_c}} \right\}\ge P_1(L_c),
\end{equation}
where $P_1(L_c)= 1 - 2 \exp\{-\frac{ {n_f} {L_c}\epsilon^{2}}{\sigma^{2}}\}$. 
Denote the estimation error of $\hat s^{\sss F}$ as $\Delta_s= \bm{1}_{n_f}^\mathsf{T}(\hat{s}^{\sss F}(k_0,L_c) - s)/n_f $, then it satisfies
\begin{equation} \label{eq:expectation-variance}
\mathop {\lim }\limits_{L_c \to \infty } | \mathbb{E}[\Delta_s] | \le 2\epsilon,~\mathop {\lim }\limits_{L_c \to \infty } \mathbb{D}[\Delta_s]=\frac{\sigma^2}{2n_f},
\end{equation}
\textcolor{black}{where $\mathbb{E}[\cdot]$ and $\mathbb{D}[\cdot]$ represent the expectation and variance of a random variable, respectively. }
\end{theorem}
\begin{proof}
The proof is provided in Appendix \ref{apdix:cs-performance}. 
\end{proof}

Theorem \ref{th:cs-performance} demonstrates that, with sufficient observations over the $\epsilon$-steady pattern, the estimation accuracy of $\hat c$ is determined by $\epsilon$. 
\textcolor{black}{In other words, the confidence interval of $\hat c(k_0,L_c)$ is given as $\hat c(k_0,L_c) \in[c-\frac{4\epsilon}{\sqrt{L_c}},c+\frac{4\epsilon}{\sqrt{L_c}}]$ with probability at least $P_1(L_c)$.}
Specifically, when $L_c\to\infty$, we have with probability one that 
\begin{equation}
\mathop {\lim }\limits_{L_c \to \infty } | \Delta_c | =0 . 
\end{equation}
However, as for the estimation accuracy of $\hat s^{\sss F}$, it only achieves $\epsilon$-level accuracy in the expected sense with bounded variance. 

\begin{remark}
Note that (\ref{eq:expectation-variance}) only presents the estimation error of $\hat s^{\sss F}$ in limit form. 
It is shown in the proof of Theorem \ref{th:cs-performance} that, one has with high probability 
\begin{equation}
\! \left \{
\begin{aligned}
&| \mathbb{E}[\Delta_s]  | \!\le\! (2 \! + \!\frac{2{k_0}+1}{L_c})\epsilon, \\
&\mathbb{D}[\Delta_s]  \!= \! \frac{\sigma^2}{2n_f} \! + \!\sigma^2(\frac{1}{n_f L_c^2} \!+\! \frac{(2{k_0}+1)^2}{L_c^2} \!+\! \frac{4{k_0}+2}{L_c} ). 
\end{aligned} \right.\!\!
\end{equation}
Despite the undesired uncertainty in $\Delta_s$, one can tighten the error bound of $\hat{s}^{\sss F}(k,L_c)$ by increasing the observations. 
\end{remark}

Note that (\ref{eq:window-s}) is not an appropriate solution if the system is not in $\epsilon$-steady pattern. 
Hence, we need to judge whether the system is in $\epsilon$-steady pattern before obtaining the final $\hat c$ and $\hat s$. 
Inspired by (\ref{eq:accuracy-c}), 
we in turn deduce that $| \Delta_c | > {4\epsilon}/{\sqrt{L_c}}$ holds with high probability if the observations used are not all in $\epsilon$-steady pattern. 
Hence, we use the last $L_c$ groups of observation to obtain a benchmark estimator of $c$ by 
\begin{equation}
\hat c_b {\buildrel \Delta \over =} \hat c(k_{end}-L_c,L_c).
\end{equation}
Based on Theorem \ref{th:cs-performance}, 
if the system is in $\epsilon$-steady pattern after $k_0$, one has with probability $1 - 2 \exp\{-\frac{ {n_f} {L_c}\epsilon^{2}}{\sigma^{2}}\}$
\begin{align} \label{eq:2c-error}
| \hat c(k_0,L_c)- \hat c_b| \le &| \hat c(k_0,L_c)- c| + |c-\hat c_b | \nonumber \\
\le & {8\epsilon}/{\sqrt{L_c}}.
\end{align}
Although infinite observations are not available in practice, the upper bound in (\ref{eq:2c-error}) can be used as an empirical criterion to judge when the $\epsilon$-steady pattern is reached, given by
\begin{itemize}
\item \textbf{$\epsilon$-steady time criterion}:
\begin{align}\label{eq:criterion-ks}
{ k_{s}} \!=\! \inf &\{ k : |\hat c(k,L_c)-\hat c_b | \le {8\epsilon}/{\sqrt{L_c}} \}.
\end{align}
\end{itemize}
Once ${k_{s}}$ is obtained, the formation input parameters $c$ and $h$ are finally determined by 
\begin{equation}\label{eq:final-ch}
\left \{
\begin{aligned}
\hat c &= \hat c(k_{s},L_{s}),  \\
\hat h^{\sss F} & =  \hat s^{\sss F}(k_{s},L_{s}) \!-\!\bm{1}_{n_f} \hat s^{i}(k_{s},L_{s}),  
\end{aligned} \right.
\end{equation}
where $L_{s}=k_{end}-k_{s}$ represents the amount of observations of the system in the $\epsilon$-steady stage.

\subsection{Range-shrink: Motivated by Truncated Estimator}

To explicitly illustrate the necessity of the range-shrink strategy, we begin with the case where the observations are noise-free and the input is known. 
Under full observation, denote $z_{k+1}^{u} {\buildrel \Delta \over =} z_{k+1}-u_{k}=Wz_{k}$. 
{\color{black}{
Then, the global topology can be obtained from $K$ groups of noise-free observations by 
\begin{equation} \label{eq:ideal-estimator}
W = Z_{2:K+1}^{u} {Z_{1:K}^\mathsf{T}}(Z_{1:K} Z_{1:K}^\mathsf{T} )^{-1} ,
\end{equation}
where $Z_{2:K+1}^{u}=[z_{2}^{u},z_{3}^{u},\cdots,z_{K+1}^{u}]$ and $Z_{1:K}=[z_{1},z_{2},\cdots,z_{K}]$. 
Note that the feasibility of the estimator under full observations relies on the invertibility of $(Z_{1:K} Z_{1:K}^\mathsf{T})$, which is related to the number of observations and the steady pattern of the formation. 
Here we temporarily suppose the invertibility holds, and analyze the details in the proposed local topology estimator in Section \ref{subsec:estimator}. 

Let $W_{\sss FF}=[w_{ij},~i,j\in\mathcal{V}_{\sss F}]\!\in\!\mathbb{R}^{n_f\times n_f}$ be the topology matrix of $\mathcal{V}_{\sss F}$. 
To infer $W_{\sss FF}$ from $\{ z_{k}^{\sss F}\}$, 
it is certainly free for one to adopt a truncated form of (\ref{eq:ideal-estimator}) as in \cite{santos2019local}
\begin{equation} \label{eq:truncated-estimator}
\hat W_{\sss FF} = Z_{2:K+1}^{u,\sss F} (Z_{1:K}^{\sss F})^\mathsf{T}(Z_{1:K}^{\sss F} (Z_{1:K}^{\sss F})^\mathsf{T} )^{-1}. 
\end{equation}
The works \cite{matta2018consistent,santos2019local,cirillo2021learning} have explored the conditions of using the truncated estimator to approximate the ground truth\footnote{In \cite{matta2018consistent,santos2019local,cirillo2021learning}, 
the conditions of using estimator (\ref{eq:truncated-estimator}) are summarized as: 
i) the topology is in symmetric Erd\H{o}s-R{\'e}nyi random graph form with vanishing connection probability, 
and ii) the ratio of the observable nodes to all nodes converges to constant as the size of the network goes to infinity.}.
Nevertheless, these conditions are not consistent with our problem setting,}} 
and $\hat W_{\sss FF}$ is far away from the ground truth from basic linear algebra, i.e., 
\begin{equation}\label{eq:unequal}
\hat W_{\sss FF} \neq  [Z_{2:K+1}^{u} {Z_{1:K}^\mathsf{T}}(Z_{1:K} Z_{1:K}^\mathsf{T} )^{-1}]_{\sss FF}. 
\end{equation}
More precisely, let $\mathcal{V}_{\sss F'}=\mathcal{V}\backslash\mathcal{V}_{\sss F}$ and the formation dynamics (\ref{eq:global-system}) can be divided into 
\begin{equation}\label{eq:devide_state}
\left[ {\begin{aligned}
{z}_{k+1}^{\sss F}\\
{z}_{k+1}^{\sss F'}
\end{aligned}} \right] \!=\! \left[ {\begin{aligned}
W_{\sss FF}~W_{\sss FF'}\\
W_{\sss F'F}~W_{\sss F'F'}
\end{aligned}} \right]\left[ {\begin{aligned}
{z}_{k}^{\sss F}\\
{z}_{k}^{\sss F'}
\end{aligned}} \right] \!+ \!\left[ {\begin{aligned}
u_{k}^{\sss F}\\
u_{k}^{\sss F'}
\end{aligned}} \right],
\end{equation}
where ${z}_{k}^{\sss F'}$ is the state of $\mathcal{V}_{\sss F'}$ at time $k$. 
Substituting $\tilde{z}_{k}^{\sss F}={z}_{k}^{\sss F}+\omega_k^{\sss F}$ into (\ref{eq:devide_state}), 
the observation of $\mathcal{V}_{\sss F}$ is given by 
\begin{equation}\label{eq:local_observation}
\tilde{z}_{k+1}^{\sss F}= {W_{\sss FF}} \tilde{z}_{k}^{\sss F} + u_{k}^{\sss F}  + {W_{\sss FF'}}z_{k}^{\sss F'} + \omega_{k+1}^{\sss F}-{W_{\sss FF}}\omega_{k}^{\sss F}.
\end{equation}
Note that (\ref{eq:local_observation}) only represents the explicit relationship of every two consecutive observations, not a real process. 
{\color{black}{
It is clear that the unobserved and non-negligible term $\{  {W_{\sss FF'}}z_{k}^{\sss F'} \}$ incurs the inequality of \eqref{eq:unequal}, 
making it extremely hard to obtain an unbiased estimator of $W_{\sss FF}$ from noisy $\{\tilde z_{k}^{\sss F} \}$. 
}}

Thanks to the constrained interaction characteristics of MRNs, 
\textcolor{black}{we observe that the robots that are outside $r_i$'s interaction range have no influence on $r_i$. }
Therefore, we transform the inference objective by shrinking the inference scope from $\mathcal{V}_{\sss F}$ to a smaller $\mathcal{V}_{\sss H}$, which directly avoids the inference bias in the truncated estimator (\ref{eq:truncated-estimator}). 
As shown in Fig.~\ref{fig-range}, we use a concentric circle to cover the feasible subset $\mathcal{V}_{\sss H} \!\subseteq\! \mathcal{V}_{\sss F}$ with radius $R_h$, satisfying
\begin{equation}
R_h=R_f-R_c.
\end{equation}  
Once the subset $\mathcal{V}_{\sss H}$ is determined, 
\textcolor{black}{we can design an unbiased estimator of the following local topology,
\begin{equation}
W_{\sss HF}\!=\![w_{ij},~i\!\in\!\mathcal{V}_{\sss H},j\in\mathcal{V}_{\sss F}]\!\in\!\mathbb{R}^{n_h\times n_f},
\end{equation}
where $n_h=|\mathcal{V}_{\sss H} |$. 
Note that $W_{\sss HF}$ covers all connections within $\mathcal{V}_{\sss H}$ and the directed connections from $\{\mathcal{V}_{\sss F}\backslash\mathcal{V}_{\sss H}\}$ to $\mathcal{V}_{\sss H}$. 
The details are presented in the next section. }

\begin{figure}[t]
\centering
\setlength{\abovecaptionskip}{0.2cm}
\includegraphics[width=0.45\textwidth]{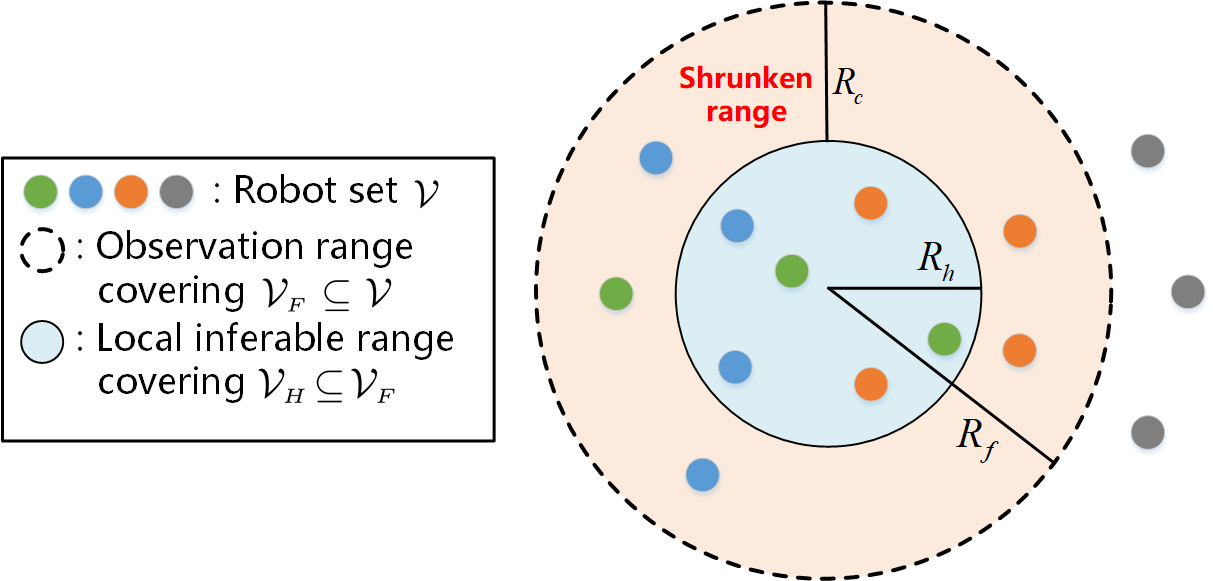}
\caption{Illustration of observation ranges. 
The blue circle area enclosing $\mathcal{V}_{\sss H}$ is with radius $R_h$, and larger circle area enclosing $\mathcal{V}_{\sss F}$ is with radius $R_f$.}
\label{fig-range}
\vspace*{-8pt}
\end{figure}

\subsection{Inferring Interaction Range by Active Excitation}\label{sec:excitation}
Next, we present the active excitation based method to illustrate how to estimate the interaction radius $R_c$.

Note that robots are equipped with sensors to detect obstacles around. 
When $r_a$ is very close to $r_j$, by the obstacle-avoidance rule (\ref{eq:obstacle-rule}), an excitation input will be triggered in $r_j$. 
Then, the observed state of $r_i$ under excitation is given by 
\begin{itemize}
\item \textbf{Observation under excitation}: 
\begin{align}\label{eq:excited-observation}
\tilde z_{k}^{j,e} = ck+s^{j} + \bm{\epsilon}_{k}^{j}+ \omega_{k}^{j} + u_{k-1}^{j,e},
\end{align}
\end{itemize}
\textcolor{black}{where $\bm{\epsilon}_{k}^{j}$ is the $j$-th element of $\bm{\epsilon}_{k}=z_k-{c} k \bm{1}_n -s$, which represents the residual error vector with the linear steady trajectory. }
According to Lemma \ref{le:ep-convergence}, when the MRN is in $\epsilon$-steady pattern, $ \| \bm{\epsilon}_{k} \|_1\le \epsilon$. 
Next, we present the details of the active excitation based method as follows. 

\begin{itemize}
\item \textit{Step 1: Initial excitation on $r_j$.}
\end{itemize}

\textcolor{black}{Based on (\ref{eq:excited-observation}) and recalling the velocity estimation error $\Delta_c$, the velocity prediction error on $j\in\mathcal{V}$ at the $\epsilon$-steady pattern is calculated by} 
\begin{align}\label{eq:excitation-j}
\delta_{k}^{j} &=\tilde z_{k}^{j,e} - \tilde z_{k-1}^{j} - \hat c \nonumber \\
&= (\omega_{k}^{j}- \omega_{k-1}^{j} ) + ( \bm{\epsilon}_k^{j}-\bm{\epsilon}_{k-1}^{j}) - \Delta_c + u_{k-1}^{j,e}.
\end{align}
Note that $\delta_{k}^{c}$ is a random variable, and $u_{k-1}^{j,e}\!=\!0$ if $r_j$ is not excited. 
Based on Theorem \ref{th:cs-performance}, if $r_j$ is under no excitation, \textcolor{black}{then we have $|\delta_{k}^{j}|\!\le\! \sqrt{3\epsilon^2\!+\!2\sigma^2}$ with a high probability. }
Utilizing this empirical result, we design the following criterion to determine whether $r_j$ is excited by $r_a$ and its reaction range (i.e., the obstacle detection range $R_o$), given by
\begin{equation} \label{eq:estimator-Ro}
 \left \{  
\begin{aligned}
k_e&=\inf\{k:| \delta_{k}^{j}  |>\sqrt{3\epsilon^2+2\sigma^2}\}, \\
\hat{R}_o &=\| {\mathbf z}_{k_e}^{j} - {\mathbf z}_{k_e}^{a} \|_2, 
\end{aligned} \right.
\end{equation}
where $k_e$ is the starting moment of the excitation stage. 

\begin{itemize}
\item \textit{Step 2: Excitation strategy.}
\end{itemize}

To keep $r_a$ within the obstacle detection range of $r_j$, we define the feasible state set of $r_a$ as  
\begin{align}
\mathcal{Z}_{k+1}^{a}= \{  {\mathbf z}_{k+1}^{a}: \| {\mathbf z}_{k+1}^{a} -{\hat{\mathbf z}}_{k+1}^{j} \|_2 \le \hat{R}_o \}, 
\end{align}
For better identification, the next movement of $r_a$ is randomly selected from $\mathcal{Z}_{k+1}^{a}$ in the same direction, i.e., 
\begin{equation}\label{eq:update-rule}
{\mathbf z}_{k+1}^{a} \!\in\! \mathcal{Z}_{k+1}^{a}\! \cap \! \{ {\mathbf z}_{k+1}^{a} : {z}_{k+1}^{a}\cdot {z}_{k}^{a} \ge 0~\text{in each dimension}\}. \!\!
\end{equation}

\begin{itemize}
\item \textit{Step 3: Estimating $R_c$ based on out-neighbors.}
\end{itemize}

If $r_j$ is injected with the excitation input $u_{k-1}^{j,e}$, the influence of $u_{k-1}^{j,e}$ will spread to $\mathcal{N}_{j}^{out}$ in following moments. 
Suppose $r_a$ makes excitations over $r_j$ for consecutive $m$ time steps, then the accumulated velocity prediction error of $i\in\mathcal{N}_{j}^{out}$ in the $m$-step is calculated by 
\begin{align}\label{eq:state-increment}
\delta_{k+m,k}^{i}= \tilde z_{k+m}^{i,e} - \tilde z_{k}^{i} - m \hat c .
\end{align}
Next, we define the following out-neighbor estimation function and demonstrate its accuracy. 
\begin{definition}[Out-neighbor indicator]\label{def:indicator}
The indicator of the event that $i\in \mathcal{N}_j^{out}$, $\Theta_{ij}$, is defined as 
\begin{equation}
\Theta_{ij}=\left \{
\begin{aligned}
&1, &&\text{if}~w_{ij}>0, \\
&0, &&\text{otherwise}. 
\end{aligned} \right.
\end{equation}
The estimator of $\Theta_{ij}$ is defined as 
\begin{equation}
\hat{\Theta}_{ij}=\left \{
\begin{aligned}
&1, &&\text{if}~|\frac{\delta_{k+m}^{i}}{m} |>(\frac{4}{\sqrt{L_c}} + \frac{4}{\sqrt{m}})\epsilon, \\
&0, &&\text{otherwise}. 
\end{aligned}\right.
\end{equation}
\end{definition}

\begin{theorem}[Accuracy of $\hat{\Theta}_{ij}$]\label{th:outneighbor-excitation}
Under $m$ consecutive excitations on $r_j$, \textcolor{black}{the true positive probability of estimator $\hat{\Theta}_{ij}$ is lower bounded as }
\begin{align}\label{eq:neighbor-iden}
&\Pr\left\{ {\Theta}_{ij}=1 |\hat{\Theta}_{ij}=1 \right\} \ge P_1(L_c)\cdot P_2(m),
\end{align}
where $P_2(m)=1 - 2 \exp\{-\frac{ {m}\epsilon^{2}}{\sigma^{2}}\}$. 
\end{theorem}
\begin{proof}
The proof is provided in Appendix \ref{apdix:outneighbor-excitation}. 
\end{proof}

Theorem \ref{th:outneighbor-excitation} demonstrates that by active excitations, the out-neighbors of $r_j$ (within the observation range) can be determined with a high probability. 
Besides, if there exists at least one out-neighbor of $r_j$ in $\mathcal{V}_{\sss F}$, then the two robots are always within the interaction range during the whole process. 
\textcolor{black}{Utilizing this characteristic, we take the maximum distance between $r_j$ and the inferred $r_j$'s out-neighbors from their observations as the lower bound of $R_h$, given by}
\begin{equation}\label{eq:range-bound}
R_c^{lb} \!=\! \sup \left\{ \| \tilde{\mathbf{z}}_{t}^{i}- \tilde{ \mathbf{z} }_{t}^{j} \|_2 : \hat{\Theta}_{ij}=1, t=1,\!\cdots\!,k_e \! +\! m \right\}. \!
\end{equation}

The procedures of obtaining $R_c^{lb}$ are summarized in Algorithm \ref{algo:infer-range}. 
Then, the interaction range satisfies 
\begin{equation}\label{eq:Rc-range}
R_c^{lb}\le R_c\le R_c^{ub}, 
\end{equation}
\textcolor{black}{where $R_c^{ub}=R_f$}. 
The range interval (\ref{eq:Rc-range}) is critical for final topology inference.

\begin{algorithm}[t]
    \caption{Infer the interaction range $R_c$}
    \label{algo:infer-range}
    \begin{algorithmic}[1]
    \REQUIRE{Steady moment $k_s$, excitation number $m$, $\hat c$ and $\hat h$.}
    \ENSURE{Lower bound of $R_c$.}
    \STATE Select a excitation target $j\in\mathcal{V}_{\sss F}$;
    \WHILE {$ | \delta_{k}^{j} |\le \sqrt{3\epsilon^2+2\sigma^2} $}
    {
       \STATE $r_a$ moves closer to $r_j$, $k=k+1$;
    }
	\ENDWHILE
	\STATE $k_e=k$, $\hat{R}_o =\| {\mathbf z}_{k_e}^{j} - {\mathbf z}_{k_e}^{a} \|_2$;
	\FOR {$t =1\to m$}
    {
    	\STATE Update ${\mathbf z}_{k_e +t}^{a}$ by (\ref{eq:update-rule});
    }
    \ENDFOR
	\FOR {all $i\in\mathcal{V}_{\sss F}\backslash\{j\}$}
	{
		\IF {$ |\frac{\delta_{k+m}^{i}}{m} |  > (\frac{4}{\sqrt{L_c}} + \frac{4}{\sqrt{m}})\epsilon $}
		{
			\STATE $\hat{\Theta}_{ij}=1$;
		}
	  	\ENDIF
    }
	\ENDFOR
	\IF { all $\hat{\Theta}_{ij}=0$}
	{
		\STATE re-select a target robot and go to line 2;
	}
  	\ENDIF
	\STATE Compute $R_c^{lb} $ by (\ref{eq:range-bound});
    \end{algorithmic}
\end{algorithm}

\section{Estimator Design and Performance Analysis}\label{sec:inference-estimation}

By the methods proposed in the last section, the obtained estimators of $\hat c$, $\hat h^{F}$ and $\hat R_c^{lb}$ make the local topology inference feasible. 
\textcolor{black}{However, directly using $\hat R_c^{lb}$ to determine $\mathcal{V}_{\sss H}$ is relatively conservative.}  
In this section, we first present the estimator of local topology $W_{\sss HF}$ and leverage it to reversely approximate $R_c$. 
Then, taking the estimation error of $\hat c$ and $\hat h^{F}$ into consideration, we give the non-asymptotic error bound of $\| \hat{W}_{\sss HF} - W_{\sss HF}  \|$. 
Finally, we demonstrate how to utilize the knowledge acquired in the active excitation stage to improve further the inference performance based on $\hat{W}_{\sss HF}$. 

\subsection{Local Topology Inference under Uncertain $R_c$} \label{subsec:estimator}
First, we analyze the inference performance of the ordinary least squares estimator under different interaction range $R_c$. 
If $R_c$ is determined, the inferable subset $\mathcal{V}_{\sss H} \!\subseteq\! \mathcal{V}_{\sss F}$ is also determined by $R_h=R_f-R_c$. 
Considering the possibility that the formation leader $r_{n}\in\mathcal{V}_{\sss F}$, we need to discriminate its influence. 
Given  $\mathcal{V}_{\sss F}$ and $k_s$, if the leader $r_{n}\in \mathcal{V}_{\sss F}$, then it is identified by 
\begin{equation}
\hat r_{n}= {\arg \mathop {\min }\limits_{i} \left\{  f_{c}^{i} : f_{c}^{i} \le \frac{8\epsilon}{\sqrt{L_c}} , i\in\mathcal{V}_{\sss F} \right\}},  \label{ww1}\\
\end{equation}
where $f_{c}^{i}={ | \sum\nolimits_{k = 0}^{k_s-1} (\tilde{z}_{k+1}^{i} - \tilde{z}_{k}^{i} - \hat c_b )| } / { {k_s} }$. 
Note that if $\hat r_{n}$ is empty, it means $r_{n}\notin \mathcal{V}_{\sss F}$. 
To discriminate this situation, we define the indicative leader vector $\mathbb{I}_{\sss F}$ by 
\begin{equation}\label{eq:indicator}
\mathbb{I}^{\sss F}(i)=\left\{ {\begin{aligned}
&1, &&\text{if}~\exists i\in \mathcal{V}_{\sss F}, i=\hat r_{n},\\
&0, &&\text{otherwise}. 
\end{aligned}} \right.
\end{equation}
Next, we will illustrate how to filter the influence of the input to infer the local topology, and use it to approximate the real $R_c$. 
Let $\mathcal{V}_{\sss H'}=\mathcal{V}_{\sss F}\backslash{\mathcal{V}_{\sss H}}$ and $W_{\sss HF}=[W_{\sss HH}~W_{\sss HH'}]$. 
Define two variables of filtered $\tilde z_{k}^{\sss F}$ and organize them as 
\begin{equation}\label{eq:filtered_observation}
\begin{aligned}
x_k &= ( \tilde z_{k}^{\sss F} -\hat h^{\sss F})\in \mathbb{R}^{n_f},    \\
y_k &= ( \tilde z_{k}^{\sss H} -\hat h^{\sss H} - \hat c \mathbb{I}^{\sss H}) \in \mathbb{R}^{n_h} ,   \\
X &=[x_0,x_1,\cdots,x_{k_s-1}]\in \mathbb{R}^{n_f \times k_s},   \\
Y &=[y_1,y_2,\cdots,y_{k_s}]\in \mathbb{R}^{n_h \times k_s}.
\end{aligned}
\end{equation}
\textcolor{black}{Then, by referring to the Theorem 2 in our preliminary work \cite{lys}, we present the following local topology estimator. }
\begin{theorem}\label{th:topo-estimator}
\textcolor{black}{Given the filtered observation matrices $X$ and $Y$, and supposing ${R}_c$ is known}, 
if $|\mathcal{V}_{\sss F}| \! + \! 1\! \le k_s$,
then the optimal estimation of $W_{\sss HF}$ in the sense of least squares is
\begin{equation}\label{eq:form-solution}
\hat W_{\sss HF}= { Y } X^\mathsf{T} ( X X^\mathsf{T})^{-1}.
\end{equation}
\end{theorem}


Theorem \ref{th:topo-estimator} gives the least squares solution of $W_{\sss HF}$ when $R_c$ is known. 
The core insight is that by the range-shrink strategy, the truncated state $[Wx]^{\sss H}=W_{\sss HF} x^{\sss F}$, which perfectly avoids the influence brought by the unobservable $\mathcal{V}_{\sss F'}$. 
Although the number of feasible observations is limited in practice, Theorem \ref{th:topo-estimator} can be used as the basis for approximating $W_{\sss HF}$ from noisy observations. 

{\color{black}{
\begin{remark}
Note that the invertibility of matrix $(X X^\mathsf{T})$, i.e., $\operatorname{Rank}(X X^\mathsf{T})=n_f$, is guaranteed from two aspects: the non-steady observations and the random observation noises. 
First, the observation matrix $X$ consists of $k_s$ columns of observations before the $\epsilon$-steady pattern is converged. 
In other words, the velocities of the robots do not reach consensus and the state variations of different robots are independent, thus making $\operatorname{Rank}(X )=n_f$ holds, which is a dominant factor. 
Second, the observations $\{ \tilde{z}_{k}^{\sss F} \}$ are corrupted by independent random noises. 
Since the columns in $X$ are calculated by $x_k = \tilde z_{k}^{\sss F} -\hat h^{\sss F}$ and are independently random, according to Sard's theorem in measure theory, the matrix is full-ranked almost surely. 
The above two factors effectively avoid the ill-posedness of the proposed estimator. 
\end{remark}

}}

\subsection{Convergence of the Proposed Estimator}
Next, we focus on the convergence performance of $\hat W_{\sss HF}$ assuming $R_c$ is known. 
Taking the estimation error of $\hat c$ and $\hat h^{\sss F}$ into account, the convergence of $\hat W_{\sss HF}$ is characterized by the following result.

\begin{theorem}[Convergence of $\hat W_{\sss HF}$ with known $R_c$]\label{th:final-error}
Let $P_3(k_s)=1-2 \exp\{-(k_s+n_h)\}$ and suppose $R_c$ is known. 
With probability at least $P_1(L_c) \cdot P_3(k_s)$, 
the error of the topology estimator $\hat W_{\sss HF}(\hat R_c)$ satisfies
\begin{equation}\label{eq:ks_bound}
 \| \hat W_{\sss HF} \!-\!  W_{\sss HF}\| = \bm{O} ( \frac{1}{k_s}) + \bm{o}( \frac{1}{k_s^2}).
\end{equation}
\end{theorem}
\begin{proof}
The proof is provided in Appendix \ref{apdix:final-error}. 
\end{proof}

Theorem \ref{th:final-error} demonstrates the convergence rate of $\hat W_{\sss HF}$ in terms of $k_s$ in probability. 
Apparently, if the observations before the $\epsilon$-steady pattern are sufficient, 
then $\hat W_{\sss HF}$ will closely approximate the ground truth in a rate of $\frac{1}{k_s}$, satisfying
\begin{equation} \label{eq:final-error}
\Pr\left\{  \mathop{\lim}\limits_{L_c,k_s\to \infty }  \| \hat W_{\sss HF} -  W_{\sss HF}\| =0 \right\}=1. 
\end{equation}
\begin{remark}
\textcolor{black}{Note that since $\hat W_{\sss HF}$ is based on the estimators $\hat c$ and $\hat h^{\sss F}$, the bound of $ \| \hat W_{\sss HF} - W_{\sss HF}\|$ is also related to $\epsilon$, $\sigma$ and $L_c$. }
In the proof of Theorem \ref{th:final-error}, we show that the RHS in (\ref{eq:ks_bound}) is in fact composed of multiple factors, including $\bm{O}( \frac{\epsilon }{k_s\sqrt{L_c}})$, $\bm{O} ( \frac{\epsilon}{k_s})$, $\bm{O} ( \frac{\sigma}{k_s})$ and $\bm{o}( \frac{\sigma^2}{k_s^2})$. 
Hence, we can characterize the bound as a uniform one about $k_s$. 
It is worth noting that, although the estimation errors of $\hat c$ and $\hat h^{\sss F}$ are influenced by $\epsilon$ and $\sigma$, these parts of errors will have a slight influence on the accuracy of $\hat W_{\sss HF}$ as $k_s$ grows. 
\end{remark}

Note that there are some possible techniques to further alleviate the influence of the observation noises, e.g., by de-regularization. 
\textcolor{black}{In this method, the optimization objective is $\sum\nolimits_{k = 1}^{k_s} \| y_k^{\sss B}- {W_{\sss HF}} y_{k-1}^{\sss A} \|^2-\beta\|W_{\sss HF}\|_{F}^2$, where $\beta>0$ and the second negative term is called de-regularization term. 
Deeper investigation towards this direction will be left as future work. }

\subsection{Accuracy Analysis}
It is illustrated in Theorem \ref{th:final-error} that if the interaction range $R_c$ is known, the local topology estimator $\hat W_{\sss HF}$ converges to $ W_{\sss HF}$ asymptotically. 
However, we only have an estimation range of $R_c$, i.e., $[R_c^{lb},R_c^{ub}]$, and different $\hat R_c$ renders different cardinality of $\mathcal{V}_{\sss F}$. 
\textcolor{black}{To analyze the accuracy of the local topology inference under various $\hat R_c$, we explicitly write the local topology estimator as $\hat{W}_{\sss HF}(\hat{R}_c)$, and propose a range approximation algorithm to find appropriate $\hat R_c$. }

\textcolor{black}{First, we use the maximum range $R_c^{ub}$ to determine an auxiliary robot set $\mathcal{V}_{\sss H_0}\subseteq \mathcal{V}_{\sss F}$, which is covered by a concentric circle of $r_a$'s observation range, with radius $R_{h0}$ satisfying
\begin{equation}
R_{h0}=R_f-R_c^{ub}. 
\end{equation} 
Let $\mathcal{V}_{\sss F_0}$ be the set of robots within the concentric circle range with radius $(R_{h0}+\hat{R}_{c})$, and denote $\mathcal{V}_{\sss F'_0}=\mathcal{V}_{\sss F}\backslash\mathcal{V}_{\sss F_0}$. 
Note that here $\mathcal{V}_{\sss H_0}$ is constant and $\mathcal{V}_{\sss F_0}$ will change with $\hat{R}_{c}$. 
Apparently, we have $\mathcal{V}_{\sss H_0}\!\subseteq\!\mathcal{V}_{\sss F_0}\!\subseteq \!\mathcal{V}_{\sss F}$ and $\mathcal{V}_{\sss H_0}\cap\mathcal{V}_{\sss F'_0}=\emptyset$. }
For the robots in $\mathcal{V}_{\sss F'_0}$, $r_a$ will regard that $\hat w_{ij} (\hat{R}_c)=0,~i\in\mathcal{V}_{\sss H_0},j\in\mathcal{V}_{\sss F'_0}$. 
For the robots in $\mathcal{V}_{\sss F_0}$, $\hat W_{\sss{H_0 F_0}}(\hat{R}_c)$ is computed by the OLS estimator. 
Combining the two parts, $W_{\sss{H_0 F}}$ is estimated by 
\begin{equation}\label{eq:h0f}
\hat W_{\sss{H_0 F}}(\hat{R}_c) =\left [Y_{\sss H_0} X_{\sss F_0}^\mathsf{T} ( X_{\sss F_0} X_{\sss F_0}^\mathsf{T})^{-1}, \bm{0}_{|\mathcal{V}_{\sss H_0}|\times |\mathcal{V}_{\sss F'_0}| } \right]. 
\end{equation}
\textcolor{black}{Recall $\hat W_{\sss{H_0 F}}(\hat{R}_c)$ utilizes $k_s$ groups of observations}, and we define the following evaluation function of $\hat R_c$ to describe its influence on $\hat W_{\sss{H_0 F}}(\hat{R}_c)$
\begin{itemize}
\item \textbf{Asymptotic inference bias of $\hat W_{\sss{H_0 F}}(\hat{R}_c)$}:
\begin{equation}\label{eq:asymptotic-bias}
f_w(\hat{R}_c)=  \mathop {\lim } \limits_{ k_s \to \infty } \| \hat W_{\sss{H_0 F}}(\hat{R}_c;k_s)  - W_{\sss{H_0 F}}  \|.  
\end{equation}
\end{itemize}

\begin{theorem}[Inference bias under different $\hat{R}_c$]\label{th:decreasing-error}
The asymptotic inference bias $f_w(\hat{R}_c)$ is monotonically decreasing w.r.t. the inferred range $\hat{R}_c$ in probability, i.e., if $\hat{R}_{c1}\ge \hat{R}_{c2}$, 
\begin{equation}\label{eq:fw-le}
\Pr \left\{ f_w(\hat{R}_{c1}) \le f_w(\hat{R}_{c2}) \right\}=1. 
\end{equation}
Specifically, if $\hat{R}_c\ge R_c$, \textcolor{black}{the estimator $\hat W_{\sss{H_0 F}}(\hat{R}_c)$ is asymptotically unbiased}, i.e., 
\begin{equation}
\Pr \left\{ f_w(\hat{R}_c)=0 \right\}=1. 
\end{equation}
\end{theorem}

\begin{proof}
The proof is provided in Appendix \ref{apdix:decreasing-error}. 
\end{proof}

Theorem \ref{th:decreasing-error} demonstrates the decreasing monotonicity of $f_w(\hat{R}_c)$ in asymptotic sense. 
\textcolor{black}{Note that $\hat{R}_c\ge R_c$ is a sufficient condition to guarantee an asymptotically unbiased $\hat W_{\sss{H_0 F}}(\hat{R}_c)$. }
Despite not knowing the groundtruth $R_c$ and $W_{\sss{H_0 F}}$ in practice, from Theorem \ref{th:decreasing-error} we deduce that $f_w(R_{c}^{ub})=f_w(R_{c})=0$, which indicates that $\hat W_{\sss{H_0 F}}(R_{c}^{ub})$ can be leveraged to replace $W_{\sss{H_0 F}}(R_{c})$ for evaluation. 
Accordingly, we define the empirical bias of $\hat W_{\sss{H_0 F}}(\hat{R}_c)$ as 
\begin{itemize}
\item \textbf{Empirical inference bias of $\hat W_{\sss{H_0 F}}(\hat{R}_c)$}:
\begin{equation}\label{eq:empirical-bias}
f_e(\hat{R}_c)= \| \hat W_{\sss{H_0 F}}(\hat{R}_c)  - \hat W_{\sss{H_0 F}}(R_{c}^{ub}) \|.  
\end{equation}
\end{itemize}

\begin{algorithm}[t]
    \caption{$\text{search}\_\text{suboptimal}\_R_c(R_{c}^{ub},R_{c}^{lb},n_c,n_w)$}
    \label{algo:infer-Rc}
    \begin{algorithmic}[1]
    \REQUIRE{Range $[R_{c}^{lb},R_{c}^{ub}]$, decision threshold $\varepsilon_w$, counting number $n_c$ and stopping threshold $n_w$. }
    \ENSURE{Suboptimal estimation of $R_{c}$.}
    \STATE $\hat{R}_c=(R_{c}^{ub}+R_{c}^{lb})/2$;
    \STATE Determine the subset $\mathcal{V}_{\sss F_0}$ by $R_{f0}=R_{h0}+\hat{R}_c$;
    \STATE Compute $f_e(\hat{R}_c)$ by (\ref{eq:empirical-bias});
	\IF {$ f_e(\hat{R}_c)> \varepsilon_w$}
	{
		\STATE $R_{c}^{lb}=\hat{R}_c$, $n_c=1$;
		\STATE $\text{search}\_\text{suboptimal}\_R_c(R_{c}^{ub},R_{c}^{lb},n_c,n_w)$;
	}
  	\ELSE
  	{
		\STATE $R_{c}^{ub}=\hat{R}_c$, $n_c=n_c+1$;
		\IF {$ n_c \ge n_w $}
		{
			\STATE Return $R_{c}^{ub}$;
		}
	  	\ELSE
	  	{
	  		\STATE $\text{search}\_\text{suboptimal}\_R_c(R_{c}^{ub},R_{c}^{lb},n_c,n_w)$;
	  	}
	  	\ENDIF
  	}
  	\ENDIF
    \end{algorithmic}
\end{algorithm}

Based on (\ref{eq:empirical-bias}), we propose Algorithm \ref{algo:infer-Rc} to obtain a suboptimal estimation of ${R}_c$ from the range $[R_c^{lb},R_c^{ub}]$. 
\textcolor{black}{The key idea of the algorithm is to validate the monotonicity of $f_e(\hat{R}_c)$, and find an appropriate $\hat{R}_c$ after which $f_e(\cdot)$ remains stable. }
Specifically, the classic bisection method is used to speed up the search efficiency, and a decision threshold $\varepsilon_w$ and a stopping threshold $n_w$ are introduced to terminate the process. 
Note that the larger $\varepsilon_w$ and smaller $n_w$ are, the more conservative $\hat{R}_c$ is.

\textcolor{black}{
\begin{remark}
In previous parts, we assumed that the observation noises on each robot are i.i.d. Gaussian noises for simple analysis.  
In fact, this assumption can be relaxed on independent but non-identical cases, i.e., $\mathbb{E} \omega_t \omega_s^{\mathsf{T}}=\delta_{ts}\operatorname{diag}(\sigma_{\omega_1}^2, \sigma_{\omega_2}^2,\cdots,\sigma_{\omega_n}^2)$.  
The key insight is to adopt $\max\{\sigma_{\omega_1}^2, \sigma_{\omega_2}^2,\cdots,\sigma_{\omega_n}^2\}$ as the variance bound for all observation noises in the inference error analysis. 
Consequently, this scaling step will not affect the convergence and asymptotic accuracy of the proposed method. 
\end{remark}}

\subsection{Estimator Design with Its Improved Solution}
With the $\epsilon$-steady pattern parameter $\hat c$, $\hat{h}^{\sss F}$ and interaction range $\hat R_c$ (output of Algorithm \ref{algo:infer-Rc}) determined, we are able to design the unbiased topology estimator of $W_{\sss HF}$ with the maximum number of robots. 
Consequently, the set $\mathcal{V}_{\sss H}$ is in turn specified by $R_{h}=R_{f}+\hat R_{c}$. 
Then, the local topology $W_{\sss HF}$ is estimated by 
\begin{equation}\label{eq:final-W-hf}
\hat W_{\sss{HF}}(\hat R_{c}) = Y(\hat R_c) X^\mathsf{T} ( X X^\mathsf{T})^{-1}. 
\end{equation}
Despite the asymptotic boundedness of the OLS estimator (\ref{eq:final-W-hf}), the proposed method nevertheless can be used as the basis for inferring the local topology when a finite number of observations are available. 

Note that (\ref{eq:final-W-hf}) only utilizes $\hat{R}_c$ to specify the inference scope of $\mathcal{V}_{\sss H}$. 
In fact, $\hat{R}_c$ can be regarded as the prior knowledge that $r_a$ has mastered in the excitation stage to further improve the inference accuracy.  
The key insight is that two robots that are not within range $R_c$ will not receive information from each other.  
Leveraging this as a hard constraint, $\hat W_{\sss HF}$ can be further optimized by solving the following problem 
\begin{subequations} \label{further}
\begin{align}
\mathop {\min }\limits_{{W_{\sss HF}}} ~& \| Y(\hat R_c) - {W_{\sss HF}}X  \|_{\sss \text{Fro}}^2  \\
\text{s.t.}~~&w_{ij}=0,\text{if}~\|\tilde{\mathbf z}^i-\tilde{\mathbf z}^j\|_2>\hat R_c, i\in\mathcal{V}_{\sss H},j\in\mathcal{V}_{\sss F}.
\end{align}
\end{subequations}
Note that (\ref{further}) is a typical constrained linear least squares problem, and can be solved by many mature optimization techniques, e.g., interior-point method \cite{boyd2004convex}.

\textcolor{black}{
Finally, we briefly summarize how the local topology $W_{\sss HF}$ is inferred from noisy observations $\{ \tilde{z}_{k}^{i},i\!\in\!\mathcal{V}_{\sss F}^{a}(k)\}_{k=1}^{k_{end}}$. 
The first step is to determine the constant subset $\mathcal{V}_{\sss F}$ and estimate the input parameters from the observations in steady pattern. 
Then, the interaction range between robots is estimated. 
Utilizing the range-shrink strategy and estimated interaction range, we further determine the appropriate subset $\mathcal{V}_{\sss H}$. 
At last, the local topology is inferred by \eqref{eq:form-solution} and its improved version \eqref{further}, where the input's influence on the non-steady observations $\{ \tilde{z}_{k}^{\sss F} \}_{k=1}^{k_s}$ is filtered. }

{\color{black}{
\subsection{Extensions and Discussions}\label{subsec:extension}
Recall that the topology estimator is obtained by solving $\mathop {\min }\limits_{{W_{\sss HF}}} ~ \| Y - W_{\sss HF} X  \|_{F}^2$. 
In fact, it can be decomposed into inferring the rows of $W_{\sss HF}$ independently, i.e., solving 
\begin{equation}
\mathop {\min }\limits_{{W_{\sss HF}^{[i,:]}}} ~ \| Y^{[i,:]} - W_{\sss HF}^{[i,:]} X  \|^2,
\end{equation}
for all $i\in\mathcal{V}_{\sss H}$. 
Based on this decomposition, we demonstrate how to infer the local topology when $\mathcal{V}_{\sss F}\subseteq\mathcal{V}_{\sss F}^{a}(k)$ does not always hold. 

Note that if there exists $k< k_{end}$ such $\mathcal{V}_{\sss F}\not\subset\mathcal{V}_{\sss F}^{a}(k)$, it indicates that the observation range of $r_a$ does not cover robots in $\mathcal{V}_{\sss F}$ simultaneously. 
Let the starting time that $i\!\in\!\mathcal{V}_{\sss F}$ be
\begin{equation}
k_f^i=\inf\left\{ k_{\ell}: i\in \bigcap\nolimits_{k=k_{\ell}}^{k_{end}} \mathcal{V}_{\sss F}^{a}(k)  \right\}.
\end{equation}
Next, as indicated in \eqref{further}, if $r_j$ is outside the interaction range of $r_i$, then the interaction weight $w_{ij}=0$. 
This property further relaxes the dependence on the observations of $\mathcal{V}_{\sss F}$. 
For an explicit expression, denote by $\tilde{\mathcal{V}}_{\sss F}^{i}$ the robot set that has possible influences on $r_i$, given by 
\begin{equation}
\tilde{\mathcal{V}}_{\sss F}^{i}=\{j: j\in \mathcal{V}_{\sss F}~\text{and}~\|\tilde{\mathbf z}^j_{k_{end}}-\tilde{\mathbf z}^j_{k_{end} } \|_2 \le \hat R_c \}. 
\end{equation}
Apparently, one has $i\in\tilde{\mathcal{V}}_{\sss F}^{i}\subseteq\mathcal{V}_{\sss F}$. 
Recalling the filtered observation variables defined in \eqref{eq:filtered_observation}, 
we permutate the filtered observations of $\tilde{\mathcal{V}}_{\sss F}^{i}$ and the local topology matrix associated with $r_i$ as follows, 
\begin{align}
\tilde{X}_i&=[\tilde{x}_{k_f^i}(i),\tilde{x}_{k_f^i+1}(i),\cdots,\tilde{x}_{k_s-1}(i)] \in \mathbb{R}^{ |\tilde{\mathcal{V}}_{\sss F}^{i}| \times (k_s-k_f^i)}~, \nonumber \\
\tilde{Y}_i&=[{y}_{k_f^i+1}^i,{y}_{k_f^i+1}^i,\cdots,{y}_{k_s}^i] \in \mathbb{R}^{ 1 \times (k_s-k_f^i)}, \nonumber \\
\tilde{W}_{i}&= [ ({ w_{ij}, {j\in\tilde{\mathcal{V}}_{\sss F}^{i}}})_{1\times |\tilde{\mathcal{V}}_{\sss F}^{i}| } , \bm{0}_{1\times (n_f-|\tilde{\mathcal{V}}_{\sss F}^{i}|) } ] \in \mathbb{R}^{ 1 \times |\tilde{\mathcal{V}}_{\sss F}^{i}|}, \nonumber
\end{align}
where $\tilde{x}_k(i)=[x_k^{\ell}]_{\ell\in\tilde{\mathcal{V}}_{\sss F}^{i}}$ is the partitioned part in $x_k$ that corresponds to $\tilde{\mathcal{V}}_{\sss F}^{i}$. 
Based on the above formation, 
the following result is presented to illustrate how to infer the local topology ${W}_{\sss HF}$ row-by-row. 
\begin{corollary}\label{coro:extension}
Given the observations before $\epsilon$-steady time $k_s$. 
For $r_i$, if $|\tilde{\mathcal{V}}_{\sss F}^{i}| \le k_s-k_f^i$, then its associated local topology $\tilde{W}_{i}$ can be uniquely inferred by 
\begin{equation}
\hat{\tilde{W}}_{i}=[\tilde{Y}_i \tilde{X}_{i}^\mathsf{T} ( \tilde{X}_i \tilde{X}_i^\mathsf{T})^{-1},\bm{0}_{1\times (n_f-|\tilde{\mathcal{V}}_{\sss F}^{i}|)}]. 
\end{equation}
\end{corollary}
The proof of this corollary is the same as that of Theorem \ref{th:topo-estimator}, and the details are omitted here. 
From Corollary \ref{coro:extension}, the available observation slot for inferring $r_i$'s local topology ${\tilde{W}}_{i}$ is not necessarily the same as that of other robots in $\mathcal{V}_{\sss F}$. 
Besides, $\hat{\tilde{W}}_{i}$ is the optimal estimation of ${\tilde{W}}_{i}$ in the sense of least squares, as long as the observation slot satisfies $|\tilde{\mathcal{V}}_{\sss F}^{i}| \le k_s-k_f^i$. 
Similar to the convergence and accuracy of $\hat W_{\sss HF}$, $\hat{\tilde{W}}_{i}$ enjoys the convergence of $\bm{O} ( \frac{1}{k_s-k_f^i})$ and the asymptotical accuracy when $k_s\to\infty$. 

In summary, although the integrated estimator $\hat W_{\sss HF}$ can be unavailable if robots in $\mathcal{V}_{\sss H}$ occur in $r_a$'s observation range at different moments, 
one can still utilize Corollary \ref{coro:extension} to infer the local topology associated with each $i\in\mathcal{V}_{\sss H}$. 
Finally, the underlying $W_{\sss HF}$ is recovered by appropriately permuting the robot indexes of $\{ \hat{\tilde{W}}_{i}, i\in\mathcal{V}_{\sss H} \}$, and stacking them into one matrix row-by-row.

\begin{remark}
The proposed inference method in this paper, which take the first-order linear formation control as the entry point, also provides insights to tackle some second-order and nonlinear cases. 
Taking the second-order linear model in \cite{cao2010sampled} as an example, 
the major difference here is that the topology matrix to be inferred describes the element-to-element interaction connections between both the positions and velocities of robots. 
The proposed method can be extended to the second-order cases because the global state evolution shares the same linear form as that of the first-order case, with appropriate notations and treatments of the double dimensions for each robot. 
Besides, for a common class of nonlinear models like $ z_{k+1}^{i}=z_{k}^{i}+ \sum\nolimits_{j = 1}^{n} w_{ij}\varphi_{ij}( z_{k}^{j}-z_{k}^{i}-h^{ij})$ (where $\varphi_{ij}(\cdot)$ is the continuous and strictly-bounded nonlinear interaction function satisfying $\varphi_{ij}(y)=0$ if $y=0$), one can still use the proposed linear estimator to infer the underlying binary adjacent topology, combined with popular clustering methods as \cite{santos2019local} does. 
\end{remark}

}}

\begin{figure}[t]
  \centering 
  \setlength{\abovecaptionskip}{0.1cm}
  \includegraphics[width=0.5\textwidth]{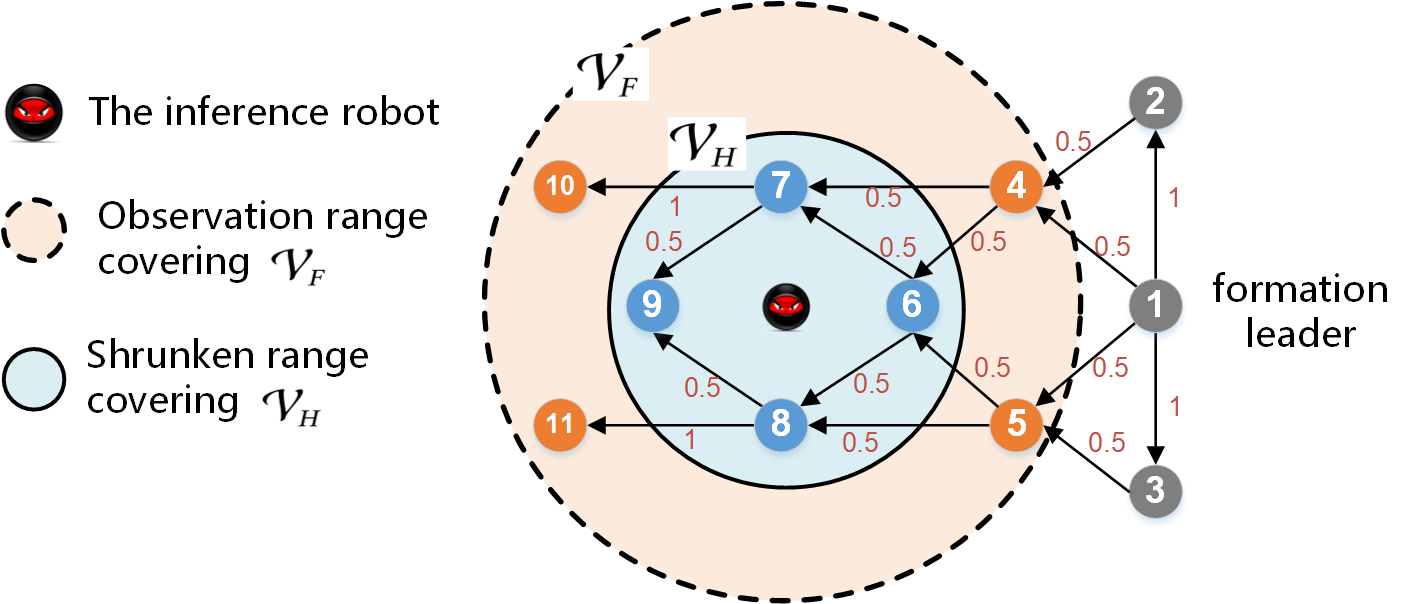} 
  \caption{An MRN of 11 robots and the interaction weights are in red font. Robot 1-3 are unobservable to $r_a$, robot 4-11 constitute the observable set $\mathcal{V}_{\sss F}$, and robot 7-9 constitute the ideal subset $\mathcal{V}_{\sss H}$.} 
  \label{se_example}
  \vspace{-10pt}
\end{figure}

\begin{figure}[t]
\centering
\setlength{\abovecaptionskip}{-0.1cm}
\subfigure[$\sigma=0.05$]{\label{fig:v_noise1}
\includegraphics[width=0.45\textwidth]{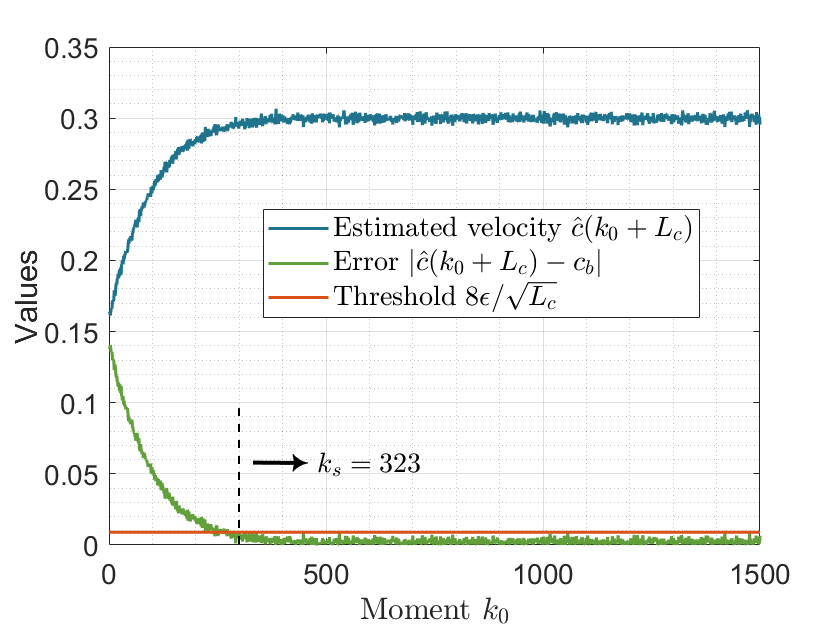}}
\hspace{-0.7cm}
\subfigure[$\sigma=0.1$]{\label{fig:v_noise2}
\includegraphics[width=0.45\textwidth]{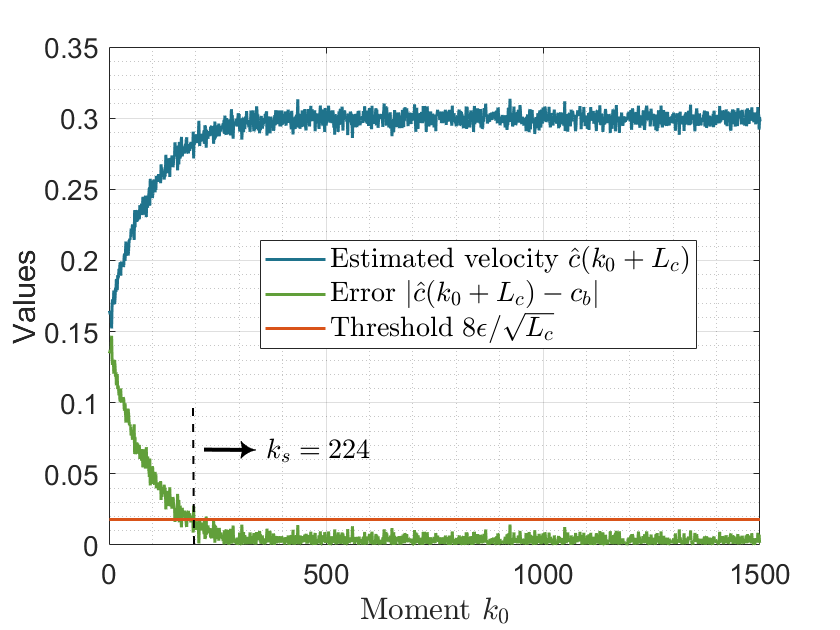}}
\caption{Estimation of the formation speed $\hat c(k,L_c) $. The threshold parameter $\epsilon$ is set as $\epsilon=0.8\sigma$. 
}
\label{fig:v_noise}
\vspace{-10pt}
\end{figure}

\begin{figure*}[t]
\centering
\subfigure[Empirical error $f_e(\hat{R}_c)$]{\label{fig:fe_error_noise}
\includegraphics[width=0.45\textwidth]{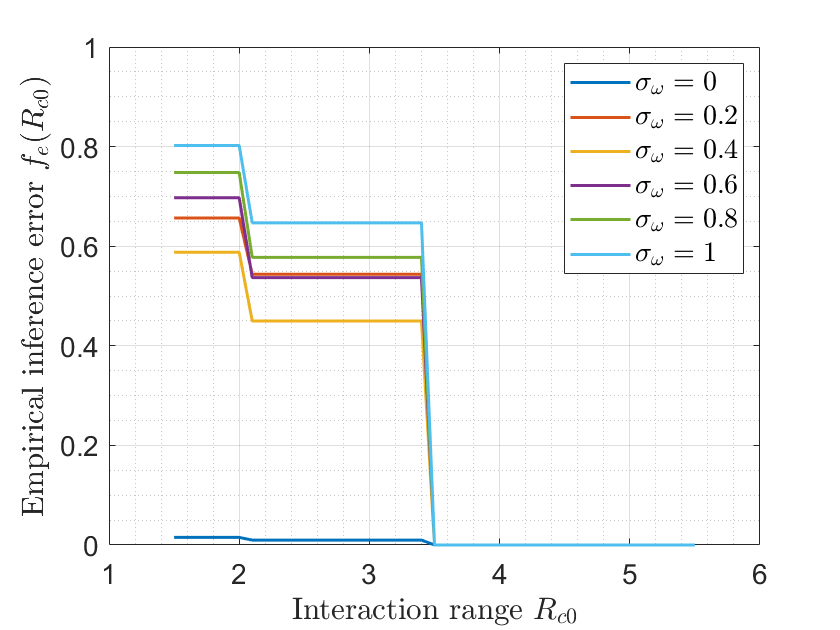}}
\hspace{-0.6cm}
\subfigure[Average error with ground truth]{\label{fig:average_error_noise}
\includegraphics[width=0.45\textwidth]{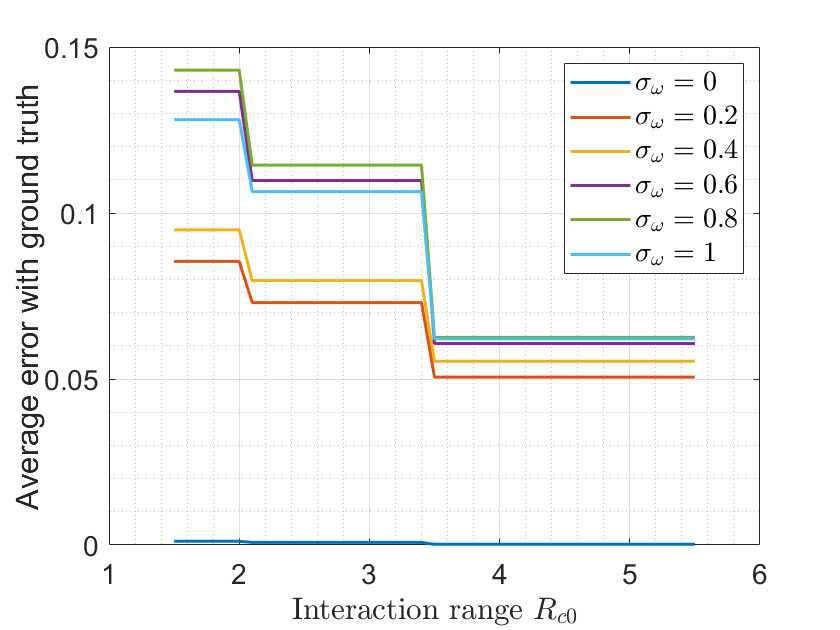}}
\hspace{-0.6cm}
\subfigure[Empirical error $f_e(\hat{R}_c)$]{\label{fig:fe_error_data}
\includegraphics[width=0.45\textwidth]{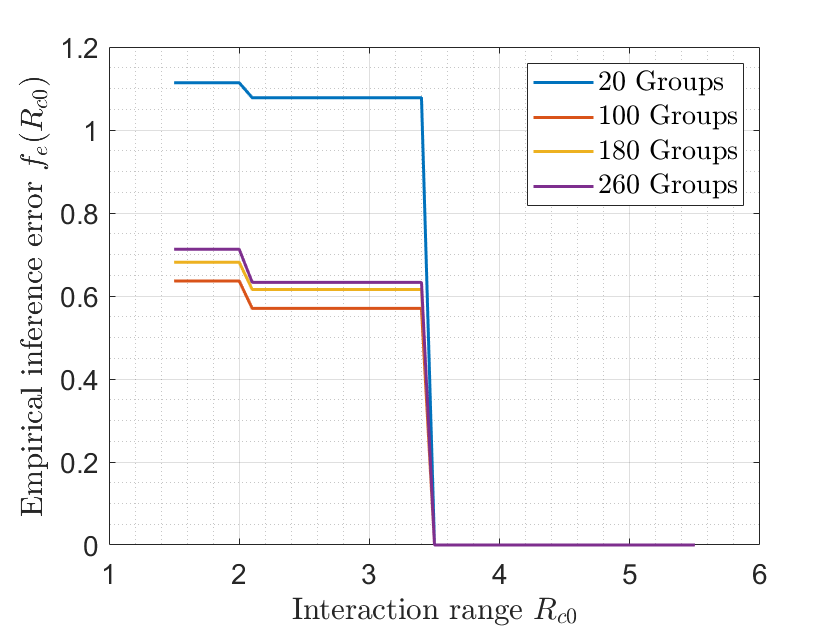}}
\hspace{-0.6cm}
\subfigure[Average error with ground truth]{\label{fig:average_error_data}
\includegraphics[width=0.45\textwidth]{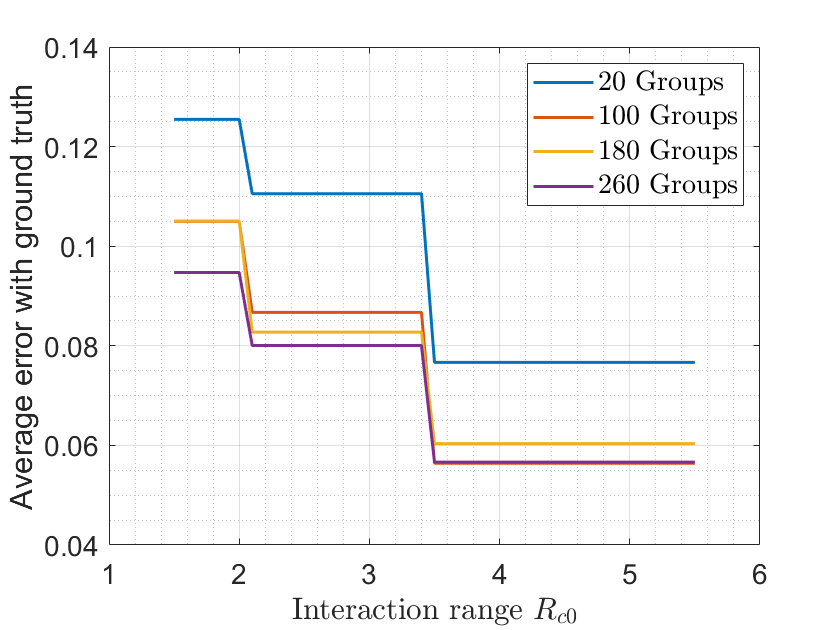}}
\vspace{-5pt}
\caption{Inference error of $\hat W_{\sss{H_0 F}}(\hat{R}_c)$. 
(a)(b): under different noise variance using 200 observations. (c)(d): under different observation amount with $\sigma=0.4$. }
\label{fig:inference_local}
\vspace*{-13pt}
\end{figure*}

\begin{figure}[t]
\centering
\subfigure[Under different noise]{\label{fig:final_noise}
\includegraphics[width=0.45\textwidth]{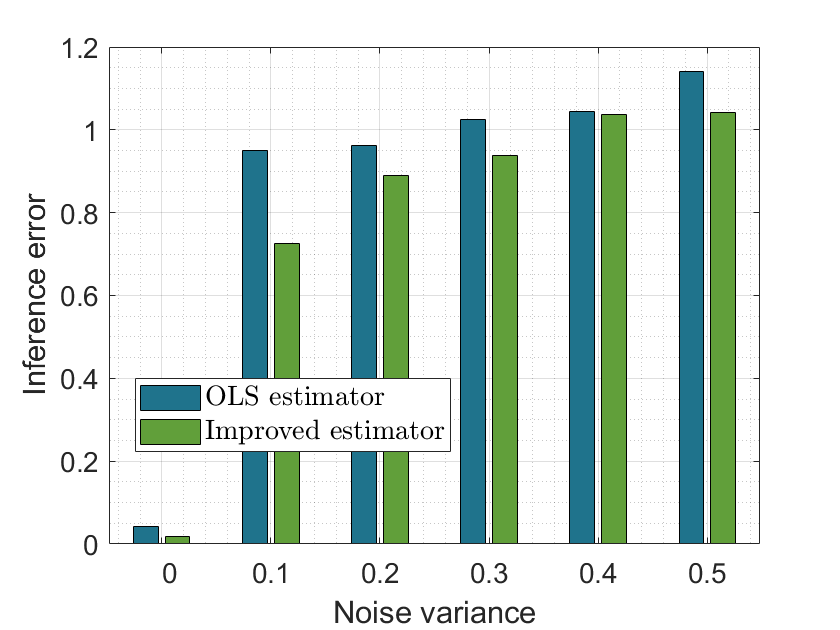}}
\hspace{-0.6cm}
\subfigure[Under different observation amount]{\label{fig:final_data}
\includegraphics[width=0.45\textwidth]{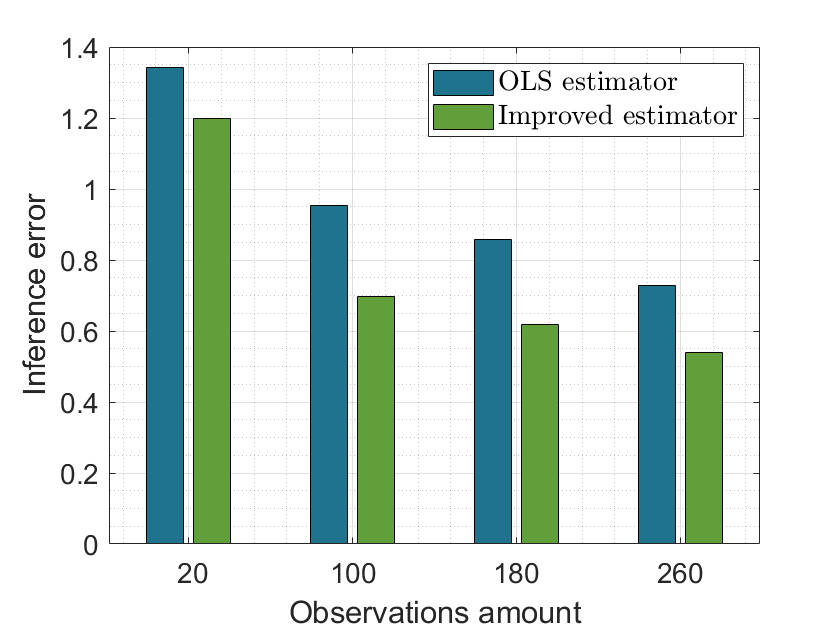}}
\vspace{-5pt}
\caption{Comparisons of inference error of OLS estimator $\hat W_{\sss{HF}}(\hat R_{c})$ and the improved estimator. 200 observations are used in (a), and $\sigma=0.1$ in (b).}
\label{fig:final_inference}
\vspace{-13pt}
\end{figure}

\section{Simulation}\label{simulation}
\subsection{Simulation Setting}
In this section, we conduct numerical experiments to demonstrate the feasibility of inferring the local topology of the MRN, and validate the theoretical results. 
For simplicity, we consider a representative case of an MRN consisting of 11 robots. 
The preset formation shape and the robot set division are shown in Fig.~\ref{se_example}. 
Specifically, robot 1 is set as the leader, and the moving speed in stable stage is set as 0.3m$/$s. 
The observation range of $r_a$ is set as $R_f=9$m, and the interaction and obstacle-detection radius of a formation robot are setting as $R_c=5$m and $R_o=1.5$m, respectively. 
The observation window length $L_c$ is set to be 500. 
In the following parts, we will present the inference results of the steady pattern, the interaction range and the local topology.

\subsection{Simulation Results}
Let us begin with examining the steady pattern estimator (\ref{eq:window-s}) in terms of $\hat c(k,L_c)$. 
For simplicity, we set the threshold parameter $\epsilon=0.8\sigma$, and conduct two groups of experiments using different $\sigma$. 
As shown in Fig.~\ref{fig:v_noise}, when the MRN reaches the steady state, the velocity estimation remains stable. 
Specifically, the red line illustrates how to find the $\epsilon$-steady time by the bound ${8\epsilon}/{\sqrt{L_c}}$ in (\ref{eq:criterion-ks}). 
Apparently, the larger $\epsilon$ is set, the more conservative $\hat c(k,L_c)$ is.

Next, as shown in Fig.~\ref{fig:inference_local}, the inference performance of the interaction range is evaluated. 
Since the active excitation method mainly aims to obtain a lower bound of $\hat{R}_c$, here we omit the simulation process of this stage and directly present the inference error under different $\hat{R}_c$. 
For fair comparisons, Fig.~\ref{fig:fe_error_noise} and \ref{fig:average_error_noise} depict the error curve under noise variance from $0$ to $1$ using 200 observations, while Fig.~\ref{fig:fe_error_data}-\ref{fig:average_error_data} depicts the error curve under observation amount from $20$ to $260$ with $\sigma=0.4$. 
Note that the average error of the inferred topology with the ground truth is computed as $ \| \hat W_{\sss{H_0 F}}(\hat{R}_c)  - W_{\sss{H_0 F}} \|/(|\mathcal{V}_{\sss{H}}||\mathcal{V}_{\sss{F}}|)$. 
As we can see, the empirical inference function $f_e(\hat{R}_c)$ and the average error is generally decreasing with $\hat{R}_c$. 
This corresponds to the result of Theorem \ref{th:decreasing-error} and supports the feasibility of using Algorithm \ref{algo:infer-Rc} to determine a more accurate $\hat R_c$. 
Specifically, the more observations are involved, the smaller the average inference error w.r.t. the ground truth is.

\textcolor{black}{
Then, with $\hat{R}_c=4.5$m, we compare the inference performance without and with the interaction constraint, corresponding to $\hat W_{\sss{HF}}(\hat R_{c})$ and \eqref{further}, respectively. }
Fig.~\ref{fig:final_noise} presents the inference errors under different variances of observation noises, varying from 0 to 0.5. 
Each test is based on 200 observations over the same time window.  
Fig.~\ref{fig:final_data} presents the inference errors under different number of observations, ranging from 20 to 260, with $\sigma=0.1$. 
Note that all the error indexes (Y-coordinate) in the figures describe the absolute deviation between two variables instead of the relative. 
\textcolor{black}{Under the same observation amount, the larger $\sigma$ is, the less improvement can be obtained. 
When $k_s$ and $\sigma$ are not very large, remarkable improvements in the inference performance can be achieved.} 
In addition, as shown in Fig.~\ref{fig:final_data}, the inference error can be further reduced with a larger available observation amount (i.e., $k_{s}$), which corresponds to the conclusion of Theorem \ref{th:final-error}. 

\textcolor{black}{
Finally, we present the performance comparison of the proposed approach with the methods in \cite{matta2018consistent} and \cite{8985069} (denoted as M-1 and M-2, respectively), under the same settings of $\hat{R}_c$, noise variance and observation amount as the last experiment. 
Note that for fair comparisons, we use the filtered observations to implement M-1. 
Fig.~\ref{fig:comp_three_noise} shows the relationships between the inference error and the observation noise variance for all methods. 
Fig.~\ref{fig:comp_three_observation} shows the relationships between the inference error and the observation amount for all methods, under common noise variance $\sigma=0.1$. 
It is clear from the presented results that the proposed method outperforms the other two, which mainly results from the estimation of the formation input and the interaction range. 
We also observe that the reason for M-1 having better performance than M-2 lies in the filtered observations we used. 
More detailed technical comparisons along with some other inference algorithms are summarized in Table \ref{tab:comparison}. 
From this table, it shows that the proposed method has better applicability for the considered problem and inference performance guarantees. 
}

\begin{figure}[t]
\centering
\subfigure[Under different noise]{\label{fig:comp_three_noise}
\includegraphics[width=0.45\textwidth]{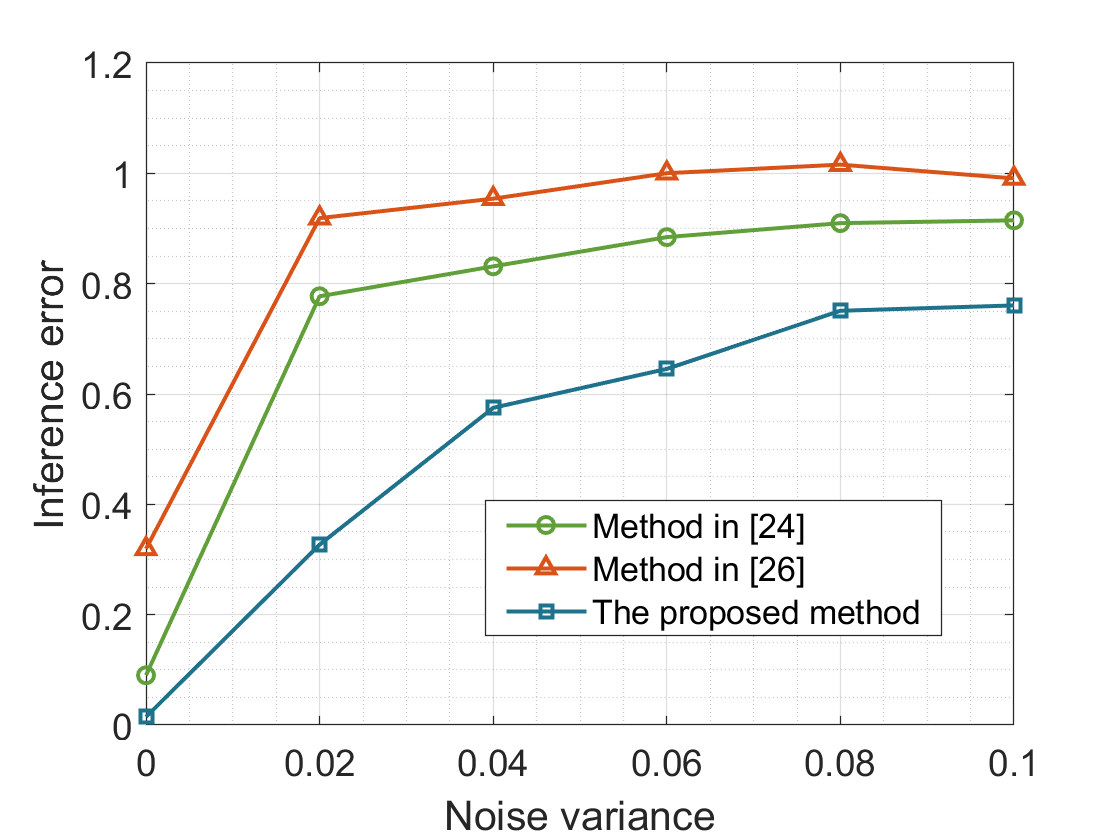}}
\hspace{-0.6cm}
\subfigure[Under different observation amount]{\label{fig:comp_three_observation}
\includegraphics[width=0.45\textwidth]{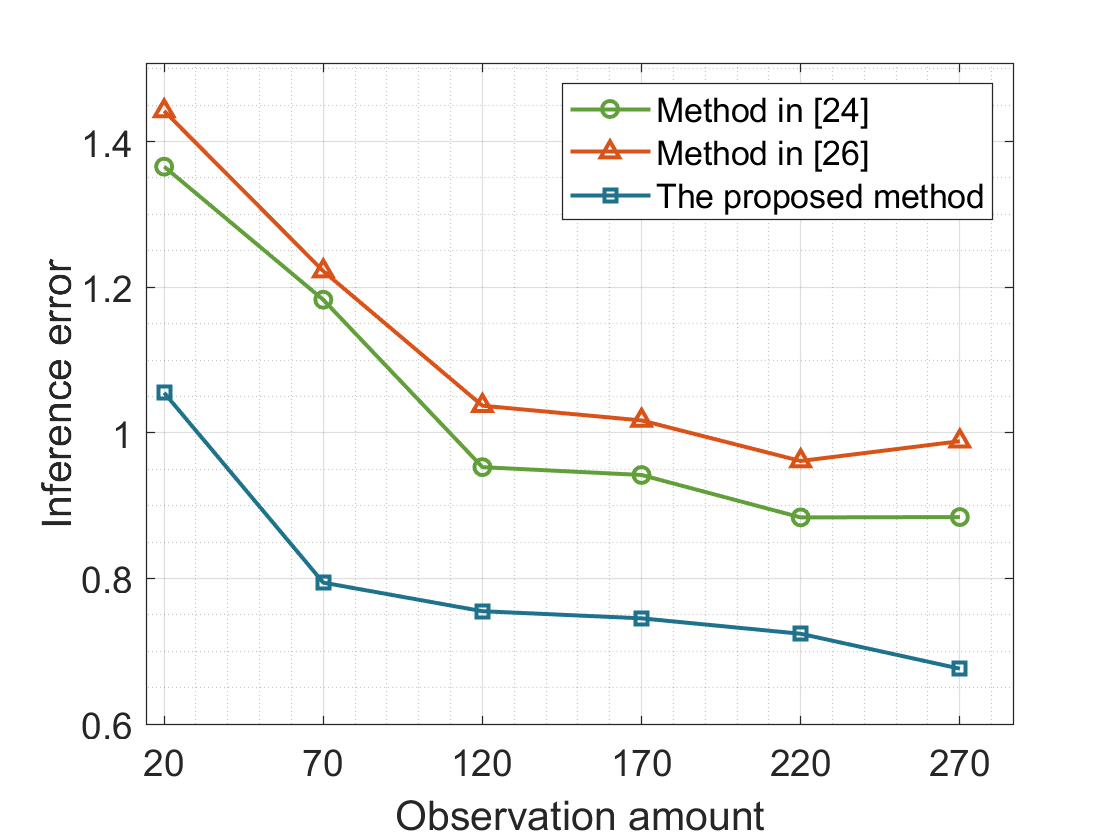}}
\vspace{-5pt}
\caption{\textcolor{black}{Comparisons of the proposed method with methods in \cite{matta2018consistent} and \cite{8985069}. 200 observations are used in (a), and $\sigma=0.1$ in (b).}}
\label{fig:compare_three}
\end{figure}


\begin{table*}[t]
\centering
\tiny
\caption{\label{tab:comparison}Comparisons of the proposed method with other methods} 
\begin{tabular}{p{1.7cm}p{1cm}p{0.9cm}p{1cm}p{1.5cm}cc p{0.9cm}p{1cm}c}
\toprule
\multirow{2}{*}{\textbf{Method}} 
& \multicolumn{2}{c}{ \textbf{Topology Structure}} 
& \multicolumn{2}{c}{ \textbf{Input Consideration} } 
& \multicolumn{2}{c}{\begin{tabular}[c]{@{}c@{}} \textbf{ Local Inference }\end{tabular}}  
& \multirow{2}{*}{\begin{tabular}[c]{@{}c@{}}\textbf{Input}\\ \textbf{Filtering}\end{tabular}} 
& \multirow{2}{*}{\begin{tabular}[c]{@{}c@{}} \textbf{Observation}\\ \textbf{Noise}\end{tabular}}    
& \multirow{2}{*}{\begin{tabular}[c]{@{}c@{}} \ \textbf{Convergence}\\ \textbf{Speed}\end{tabular}} \\
\cmidrule{2-7}
& Undirected    & Directed    & Random & Non-random  & Feasibility & Conditions \\ 
\midrule
Truncated estimator in \cite{matta2018consistent} & \makecell[c]{$\checkmark$}  &  &\makecell[c]{$\checkmark$}  &  & \makecell[c]{$\checkmark$}     &\makecell{Erd\H{o}s-R{\'e}nyi graph\\$ \frac{N_O}{N}\!\in\! (0,1] (N\!\to\!\infty)$$^{^1}$}  &   &  &$\bm{O}(\sqrt{ \frac{ 1}{T} })$$^{^2}$   \\ 
\midrule
Spectral method in \cite{zhu2020network} & \makecell[c]{$\checkmark$}     &      & \makecell[c]{$\checkmark$}  &  &   &     & &  &$\bm{O}(e^{-L})$$^{^3}$  \\
\midrule
Geometric method in \cite{vasquez2018network}  &    &\makecell[c]{$\checkmark$}  & &\makecell[c]{$\checkmark$}  &\makecell{ \textcolor{black}{Binary} \\ \textcolor{black}{judgement$^{^4}$}}   &\makecell{ Non-steady trajectory\\is available}   &  &\makecell[c]{$\checkmark$} &No guarantee\\ 
\midrule
OLS-based method in \cite{8985069}  &    &\makecell[c]{$\checkmark$}  & &\makecell[c]{$\checkmark$}  &\makecell{Feasible \\ if revised$^{^5}$}    &\makecell{\textcolor{black}{Non-steady trajectory}\\\textcolor{black}{is available}} &  &\makecell[c]{$\checkmark$} &No guarantee\\ 
\midrule
\textbf{Our method}  & & \makecell[c]{$\checkmark$}    &  & \makecell[c]{$\checkmark$}  &\makecell[c]{$\checkmark$}  &\makecell{Non-steady trajectory\\is available}  &\makecell[c]{$\checkmark$} &\makecell[c]{$\checkmark$}   &$\bm{O}(\frac{1}{T})$   \\ 
\bottomrule
\addlinespace[0.5ex]
\setlength\tabcolsep{0.5ex}
\end{tabular}
\vspace{-5pt}
\begin{tablenotes}[para]\footnotesize
    \item[1] $N_O$ and $N$ represent the cardinality of the observable subset and the whole node set, respectively. 
    \item[2] $T$ here refers to the number of observations in the non-steady stage, and the system in this reference is asymptotically stable. 
    \item[3] In \cite{zhu2020network} the authors implement $L$ groups of tests over the system, with the same initial state while ending at different moments, and no noise terms are involved.  
    \item[4] The method is based on the geometric characteristics of the robot trajectory. Although not tailored for the local topology inference of MRNs, but we point out it can applied to infer whether the connection between two robots exists. 
	\item[5] The method is originally designed for global topology inference, and can be revised for the local cases if using the idea in this paper. 
\end{tablenotes}
\end{table*}

\section{Conclusion}\label{conclusion}
\textcolor{black}{In this paper, we have studied the problem of inferring the local topology of MRNs under first-order formation control, without the knowledge about formation input and interaction parameters}. 
To overcome the inference challenges brought by the unknown formation input and interaction range, 
\textcolor{black}{we first demonstrated how to determine the available robot subset for inference, considering the set of robots that are within the observation range of the inference robot might change over time. }
Then, we designed $\epsilon$-steady pattern estimators to obtain the input parameters and an active excitation method to estimate the interaction range. 
Then, we proposed a range-shrink strategy to avoid the inference brought by the unobservable robots and presented the local topology estimator. 
The convergence rate and the accuracy of the proposed estimator were proved. 
\textcolor{black}{Extensions on different observation slots for the robots and more complicated control models were also analyzed. 
Finally, extensive simulation tests and comparisons verified the effectiveness of the proposed method.} 
Future directions include i) generalizing the method to more complex formation control cases (e.g., switching topology and nonlinear dynamics); 
ii) investigating the possible attack against the MRNs based on the inferred topology along with its countermeasures.

\appendix 

\subsection{Proof of Theorem \ref{th:cs-performance}}\label{apdix:cs-performance}
\begin{proof}
Based on Lemma \ref{le:ep-convergence}, when the MRN has reached $\epsilon$-steady pattern after $k_0$, the dynamics (\ref{eq:global-system}) is equivalent to 
\begin{align}\label{eq:equivalent-state}
z_{k}^{\sss F}=ck\bm{1}_{n_f}+s^{\sss F}+\bm{\epsilon}_k^{\sss F},
\end{align}
where $\bm{\epsilon}_k=W^{k}z_{0}+\left(\sum\nolimits_{t=0}^{k-1} W^{t}\right)u - ck\bm{1}_n-s $, satisfying $\| \bm{\epsilon}_k \|_1 \le \epsilon $. 
Then, the estimation error of $\hat c(k_0,L_c)$ is given by 
\begin{align}\label{eq:c-error0}
\Delta_c=&\hat c(k_0,L_c)-c = \sum\limits_{t= k_0}^{k_0+L_c-1}  \sum\limits_{i\in\mathcal{V}_{\sss F}}  \frac{(\tilde{z}_{t+1}^{i}-\tilde{z}_{t}^{i}-c) } { {n_f}{L_c} } \nonumber  \\
=& \sum\limits_{t= k_0}^{k_0+L_c-1}  \sum\limits_{i\in\mathcal{V}_{\sss F}} \frac{( \bm{\epsilon}_{t+1}^i - \bm{\epsilon}_{t}^i + \omega_{t+1}^i - \omega_{t}^i  ) } { {n_f}{L_c} } \nonumber \\
=& \sum\limits_{i\in\mathcal{V}_{\sss F}} \frac{( \bm{\epsilon}_{k_0+L_c}^i - \bm{\epsilon}_{k_0}^i + \omega_{k_0+L_c}^i - \omega_{k_0}^i  ) }{ {n_f}{L_c} } . 
\end{align}
For ease notation, define
\begin{equation}\label{eq:epsilon-omega-c}
\left \{
\begin{aligned}
\bar{\epsilon}_c&=\sum\limits_{i\in\mathcal{V}_{\sss F}} ( \bm{\epsilon}_{k_0+L_c}^i - \bm{\epsilon}_{k_0}^i  )/{ ({n_f}{L_c}) }, \\
\bar{\omega}_c&=\sum\limits_{i\in\mathcal{V}_{\sss F}} (\omega_{k_0+L_c}^i - \omega_{k_0}^i  )/{ ({n_f}{L_c}) },  
\end{aligned} \right.
\end{equation}
Then, we only need to prove the upper bound of $\Delta_c = \bar{\omega}_c + \bar{\epsilon}_c . $

Note that $\omega_k^i\sim N(0,\sigma^2)$ i.i.d., and thus $\bar{\omega}_c$ is subject to $N(0,2 \sigma^2/({n_f} {L_c^2}))$. 
{\color{black}{
Based the Gaussian nature, one can obtain the following concentration inequality by finding the optimal Chernoff bound of $\bar{\omega}_c$ (see \textit{example 2.1} in \cite{wainwright2019high}), i.e.,
}}
\begin{equation}\label{eq:wc1}
\Pr\{ |\bar{\omega}_c| \le \frac{2 \epsilon}{\sqrt{L_c}} \} \ge 1 - 2 \exp\{-\frac{{n_f} {L_c}\epsilon^{2}}{ \sigma^{2} }\}=P_1(L_c).
\end{equation}
As for $|\bar{\epsilon}_c| $, since the system is in $\epsilon$-steady pattern, one has 
\begin{equation}\label{eq:epsilonc1}
|\bar{\epsilon}_c|\le 2 \epsilon/L_c \le 2 \epsilon/\sqrt{L_c}. 
\end{equation}
Hence, combining (\ref{eq:wc1}) and (\ref{eq:epsilonc1}) yields that 
\begin{align}\label{eq:accuracy-c0}
\Pr\left \{ |\bar{\epsilon}_c|+|\bar{\omega}_c| \le \frac{4\epsilon}{\sqrt{L_c}} \right\} \!\ge\! P_1(L_c).  
\end{align}
Substituting the absolute inequality $|\Delta_c| \le |\bar{\omega}_c| + |\bar{\epsilon}_c|$ into (\ref{eq:accuracy-c0}), the first statement (\ref{eq:accuracy-c}) is proved. 

Next, consider the estimation error of $\hat{s}^{\sss F}(k;L_c)$. 
Based on (\ref{eq:equivalent-state}), one has 
\begin{equation}\label{eq:bound-s}
\hat{s}^{\sss F}(k_0,L_c)- {s}^{\sss F} = \sum\limits_{t = k_0+1}^{k_0+L_c} (\bm{\epsilon}_t^{\sss F} + \omega_t^{\sss F} - \Delta_c t \bm{1}_{n_f})/L_c,
\end{equation}
Similar with $\Delta_c$, the average estimation of all elements in (\ref{eq:bound-s}) is given by 
\begin{equation}\label{eq:s-error1}
\Delta_s= \frac{\bm{1}_{n_f}^\mathsf{T} {\bm\epsilon}_s }{n_f}+ \frac{\bm{1}_{n_f}^\mathsf{T} {\omega}_s }{n_f} - \frac{\Delta_c(2k_0+L_c+1)}{2},
\end{equation}
where ${\bm\epsilon}_s=\sum\limits_{t = k_0+1}^{k_0+L_c} \bm{\epsilon}_t^{\sss F}/L_c$ and ${\omega}_s=\sum\limits_{t = k_0+1}^{k_0+L_c} \omega_t^{\sss F}/L_c$. 
Taking $\bar{\epsilon}_c$ and $\bar{\omega}_c $ into the third term $\frac{\Delta_c(2k_0+L_c+1)}{2}$, one obtains 
\begin{align}\label{eq:s-error2}
\frac{\Delta_c(2k+L_c+1)}{2}=(\bar{\epsilon}_c + \bar{\omega}_c )(2k_0+L_c+1)/2 . 
\end{align}
Substituting (\ref{eq:s-error2}) into (\ref{eq:s-error1}), then one has $\Delta_s= \bar{\epsilon}_s + \bar{\omega}_s$ where 
\begin{equation}\label{eq:epsilon-omega-s}
\left \{
\begin{aligned}
\bar{\epsilon}_s &= \frac{ \bm{1}_{n_f}^\mathsf{T}  {\bm\epsilon}_s }{n_f} + \frac{(2k_0 +L_c+1)\bar{\epsilon}_c}{2}, \nonumber \\
\bar{\omega}_s &= \frac{ \bm{1}_{n_f}^\mathsf{T}  {\omega}_s  }{n_f} + \frac{(2k_0 +L_c+1)\bar{\omega}_c}{2}. \nonumber
\end{aligned} \right.
\end{equation}
Utilizing $|\bar{\epsilon}_c|\le 2 \epsilon/L_c$ and $\omega_k^i\sim N(0,\sigma^2)$, one deduces that 
\begin{align} | \bar{\epsilon}_s | &\le (2+\frac{2k_0+1}{L_c})\epsilon, \label{eq:final-error-s-1}
 \\
\mathbb{D}[\bar{\omega}_s] & \!=\! \frac{\sigma^2}{n_f L_c^2} \!+\! \frac{\sigma^2}{2n_f}(1+ \frac{(2k_0+1)^2}{L_c^2}+ \frac{4k_0+2}{L_c} ). \label{eq:final-error-s-2}
\end{align}
Combining (\ref{eq:final-error-s-1}) and (\ref{eq:final-error-s-2}), \textcolor{black}{it yields that $\mathbb{E}[ | \bar{\epsilon}_s | + \bar{\omega}_s ]=|\bar{\epsilon}_c|$ and $\mathbb{D}[| \bar{\epsilon}_s | + \bar{\omega}_s ]=\mathbb{E}[\bar{\omega}_s ]$ holds.} 
When $L_c \to \infty $, the second statement (\ref{eq:expectation-variance}) is proved. 
\end{proof}

\subsection{Proof of Theorem \ref{th:outneighbor-excitation}}\label{apdix:outneighbor-excitation}
\begin{proof}
To begin with, similar with (\ref{eq:excitation-j}), the accumulated velocity prediction error $\delta_{k+m,k}^{i}$ is expanded as 
\begin{align}\label{eq:velocity-increment}
\delta_{k+m,k}^{i}= &\sum\limits_{t = 1}^m {w_{ij}^{(t)}u_{k + m - t-1}^{j,e} } - m \Delta_c  \nonumber \\ 
&+{\omega_{k+m}^{i}- \omega_{k+1}^{i}} + {\bm{\epsilon}_{k+m}^{i}-\bm{\epsilon}_{k+1}^{i}}  ,
\end{align}
where ${w_{ij}^{(t)}}$ represents the $t$-power of $w_{ij}$.

Then, it follows that the average velocity prediction error in the $m$-period is given by 
\begin{align}\label{eq:velocity-increment2}
\frac{\delta_{k+m,k}^{i}}{m}= & \frac{1}{m} \sum\limits_{t = 1}^m {w_{ij}^{(t)}u_{k + m - t-1}^{j,e} } - \Delta_c  \nonumber \\ 
&+\frac{ {\omega_{k+m}^{i}- \omega_{k+1}^{i}} }{m} + \frac{\bm{\epsilon}_{k+m}^{i}-\bm{\epsilon}_{k+1}^{i}}{m}. 
\end{align}

Note that $\frac{ {\omega_{k+m}^{i}- \omega_{k+1}^{i}} }{m}$ is subject to $N(0,2\sigma^2/m^2)$. 
Applying the Chernoff concentration inequality, one obtains
\begin{align}
\Pr\left\{ \frac{| {\omega_{k+m}^{i}- \omega_{k+1}^{i}} | }{m}  \!\le\! \frac{2\epsilon}{\sqrt{m}} \right\} \!\ge\! 1 \!-\! 2 \exp\{-\frac{ {m}\epsilon^{2}}{\sigma^{2}}\} \!=\! P_2(m).
\end{align}
Recall that it is proved that $\Pr\left\{ | \Delta_c | \le {4\epsilon }/{\sqrt{L_c}}\right\}\ge P_1(L_c) $ in Theorem \ref{th:cs-performance}, and $ |\frac{\bm{\epsilon}_{k+m}^{i}-\bm{\epsilon}_{k+1}^{i}}{m}|\le \frac{\epsilon}{m} \le \frac{\epsilon}{\sqrt{m}}$ always holds. 
Combining these pieces, if $r_j$ is under no excitation (i.e., $u^{j,e}=0$), it yields that 
\begin{align}\label{eq:ideal-error}
&\Pr\left\{ \frac{ | \delta_{k+m}^{i} | }{m} \le (\frac{4}{\sqrt{L_c}} + \frac{4}{\sqrt{m}})\epsilon \right\} \ge P_1(L_c) \cdot P_2(m) . 
\end{align}
Taking the converse-negative version of the statement (i.e., if $u^{j,e}=0$ then (\ref{eq:ideal-error}) holds), it yields the statement described by (\ref{eq:neighbor-iden}). 
The proof is completed. 
\end{proof}

\subsection{Proof of Theorem \ref{th:final-error}}\label{apdix:final-error}
\begin{proof}
For notational brevity, let $u_{\Delta}= \hat h^{\sss H}- h^{\sss H} + (\hat c -c) \mathbb{I}^{\sss H}$ be the error vector of the filtered input. 
Then, for each pair of $(y_k,x_k)$, we have 
\begin{equation}\label{eq:y-error}
y_{k}=W_{\sss HF} x_{k}+ (W_{\sss HH} - I_{n_h}) u_{\Delta} + \omega_{k+1}^{\sss H} -W_{\sss HF} \omega_{k}^{\sss F}.  
\end{equation}
Let $U_{\Delta}=[u_{\Delta},\cdots,u_{\Delta}]\in\mathbb{R}^{n_h \times k_s}$, $\Omega^{-}=[\omega_{0},\cdots,\omega_{k_s-1}]$ and $\Omega^{+}=[\omega_{1},\cdots,\omega_{k_s}]$. 
Then, all the filtered observations can be compactly written as
\begin{equation}\label{eq:Y-error}
Y=W_{\sss HF} X+ (W_{\sss HH} - I_{n_h}) U_{\Delta} + \Omega_{\sss H}^{+} -W_{\sss HF} \Omega_{\sss F}^{-}.  
\end{equation}
Based on the estimator (\ref{eq:form-solution}), the inference error matrix of $\hat W_{\sss HF}$ is given by 
\begin{align}\label{eq:W-error}
&\| \hat W_{\sss HF} -  W_{\sss HF} \|   \nonumber  \\
=& \|((W_{\sss HH} \!-\! I_{n_h}) U_{\Delta} \!+\! \Omega_{\sss H}^{+} \!-\! W_{\sss HF} \Omega_{\sss F}^{-}) X^\mathsf{T}  (XX^\mathsf{T})^{-1} \|   \nonumber  \\
\le&  (\|(W_{\sss HH} \!-\! I_{n_h}) U_{\Delta} X^\mathsf{T} \| \! + \! \| (\Omega_{\sss H}^{+} \!-\! W_{\sss HF} \Omega_{\sss F}^{-} ) X^\mathsf{T}\| ) \| (XX^\mathsf{T})^{-1} \| .
\end{align}
Next, we will prove that each term of the RHS of (\ref{eq:W-error}) is bounded individually. 

\begin{itemize}
\item \textit{Part 1: Upper Bounding $\| (XX^\mathsf{T})^{-1} \|$. }
\end{itemize}
Recall that the system state can be expanded recursively as 
\begin{align} \label{eq:real-state}
\!\! z_{k+1}&=W^{k+1} z_{0} +  \left(\sum\nolimits_{t=0}^{k} W^{t}\right) L h + \left(\sum\nolimits_{t=0}^{k} W^{t}\right) u_{c} \nonumber \\
&= c(k+1)\bm{1}_n \!+\! W^{k+1} z_{0} + \sum\limits_{t = 0}^k (\sum\limits_{i = 2}^n {\lambda_i^t} {{q_i}{v_i^\mathsf{T}}} ) u. 
\end{align}
Note that the last two terms in the RHS of (\ref{eq:real-state}) is strictly bounded regardless of the moment $k$, and the influence of the independent observation noises will not accumulate during the state evolution. 
Since the term $c(k+1)\bm{1}_n$ takes the dominant role in the state evolution, we can directly characterize $\|X\|$ in terms of the time horizon $k_s$. 

Based on the matrix norm inequality $\|X\|\le \|X\|_{F}\le \sqrt{\min\{n_h,k_s\}}\|X\|$, one easily infers that
\begin{equation}
\frac{\|X\|_{F}}{\sqrt{\min\{n_h,k_s\}}} \le \|X\| \le \|X\|_{F}. 
\end{equation}
Further utilizing $\| z_{t} z_{t}^\mathsf{T} \|_{F}=\bm{O}( t^2)$ and $\sum\nolimits_{t=0}^{k_s} t^2= k_s(k_s+1)(2k_s+1)/6$, one deduces that $ \| \sum\nolimits_{t=0}^{k_s} z_{t} z_{t}^\mathsf{T} \|_{F}= \bm{O}( k_s^3)$. 
Then, it follows that 
\begin{align}\label{eq:ks3}
\| X \| \!=\! \bm{O}( k_s^{3/2}), \| X X^\mathsf{T}\| \!=\! \bm{O}( \| \sum\limits_{t=0}^{k_s} z_{t} z_{t}^\mathsf{T} \| ) \!=\! \bm{O}( k_s^3). 
\end{align}
By (\ref{eq:ks3}), one can always find a group of $\beta_1, \beta_2\in\mathbb{R}^{+}$, such that $\beta_1 k_s^3 I_{n} \preceq X X^\mathsf{T}  \preceq \beta_2 k_s^3 I_{n}$. 
In turn, the inverse matrix of $X X^\mathsf{T}$ satisfies
\begin{equation}
\frac{I_{n} }{\beta_1 k_s^3} \succeq (X X^\mathsf{T})^{-1} \succeq \frac{I_{n}}{\beta_1 k_s^3}, 
\end{equation}
where the invertibility of $X X^\mathsf{T}$ is guaranteed by the i.i.d. noise $\{\omega_k\}_{k=0}^{k_s}$. 
Therefore, the spectral norm of $X X^\mathsf{T}$ is characterized by 
\begin{equation}\label{eq:yy-bound}
\| (XX^\mathsf{T})^{-1}  \|=\bm{O}(\frac{1}{k_s^3}). 
\end{equation}

\begin{itemize}
\item \textit{Part 2: Upper Bounding $\| \Omega_{\sss H}^{+} - W_{\sss HF} \Omega_{\sss F}^{-} X^\mathsf{T} \|$. }
\end{itemize}
Utilizing the concentration measure in Gaussian space \cite{davidson2001local}, given a random matrix $\Omega\in \mathbb{R}^{n_h \times k_s}$ with independent standard normal entries, one has with probability at least $1-2 \exp \{-r^{2} / 2 \}$ 
\begin{equation}
\| \Omega  \| \leq \sqrt{k_s}+\sqrt{n}+r. 
\end{equation}
Let $r=\sqrt{2(k_s+n)}$ and utilize $\sqrt{k_s}+\sqrt{n}\le \sqrt{2(k_s+n)}$, one has with probability at least $1-2 \exp \left(-(k_s+n)\right)$ that $\| \Omega\| \leq 2\sqrt{2(k_s+n)} \sigma$. 
Note that the observation matrix $X$ contains the noise matrix $\Omega_{\sss F}^{-}$, i.e., $\operatorname{Exp}[\Omega_{\sss F}^{-}X^\mathsf{T}]=\sigma^2 I_{n_f}$. 
Therefore, it yields that with probability at least $P_3(k_s)=1-2 \exp \left\{-(k_s+n_h)\right\}$
\begin{align}\label{eq:oo-bound}
\| \Omega_{\sss H}^{+} - W_{\sss HF} \Omega_{\sss F}^{-} X^\mathsf{T} \| &=\bm{O}(\sigma k_s^{ 2}) + \bm{o}( \sigma^2 {k_s}). 
\end{align}

\begin{itemize}
\item \textit{Part 3: Upper Bounding $\| (W_{\sss HH} \!-\! I_{n_h}) U_{\Delta} X^\mathsf{T} \|$. }
\end{itemize}
Recall that $u_{\Delta}= \hat h^{\sss H} + \hat c \mathbb{I}^{\sss H}$, it can be split as 
\begin{equation}
U_{\Delta}=U_{\Delta}^c+U_{\Delta}^{h}, 
\end{equation}
where $ U_{\Delta}^c = [ (\hat c \!-\!c) \mathbb{I}^{\sss H},\cdots,(\hat c \!-\!c) \mathbb{I}^{\sss H}]$ and $ U_{\Delta}^H =[ \hat h^{\sss H} \!-\!h,\cdots,\hat h^{\sss H} \!-\! h]$. 
Then, the upper bound of $\| (U_{\Delta} + \Omega_{\sss H}^{+})X^\mathsf{T} \|$ is given by 
\begin{equation}
\| (U_{\Delta} + \Omega_{\sss H}^{+})X^\mathsf{T} \| \le \|U_{\Delta}^c X^\mathsf{T}\| + \|U_{\Delta}^h X^\mathsf{T}\| + \| \Omega_{\sss H}^{+}X^\mathsf{T} \|. 
\end{equation}
By the conclusion of Theorem \ref{th:cs-performance}, with probability at least $P_1(L_c)=1 - 2 \exp\{-\frac{ {n_f} {L_c}\epsilon^{2}}{\sigma^{2}}\}$, one has $\|U_{\Delta}^c\|\le 4\epsilon \sqrt{n_f k_s/L_c }$ and 
\begin{equation}\label{eq:uc}
 \|U_{\Delta}^c X^\mathsf{T} \|= \bm{O}( \frac{\epsilon k_s^2}{\sqrt{L_c}}). 
\end{equation}
Note that Theorem \ref{th:cs-performance} demonstrates that each term in $U_{\Delta}^h$ is a random variable of $\epsilon$-level mean and $\sigma^2$-level variance, and thus one deduces that with probability at least  $P_3(k_s)$
\begin{equation}\label{eq:uh}
 \|U_{\Delta}^h X^\mathsf{T} \|= \bm{O}( \epsilon k_s^2) + \bm{O}( \sigma k_s^2). 
\end{equation}
Combining (\ref{eq:uc}) and (\ref{eq:uh}), one obtains
\begin{equation}\label{eq:uuu}
 \|U_{\Delta} X^\mathsf{T} \|= \bm{O}( \epsilon k_s^2) + \bm{O}( \sigma k_s^2). 
\end{equation}

Finally, multiplying the terms (\ref{eq:oo-bound}) and (\ref{eq:uuu}) with (\ref{eq:yy-bound}), one has with probability at least $P_1(L_c)\cdot P_3(k_s)$
\begin{equation}\label{eq:4terms}
 \| \hat W_{\sss HF} \!-\!  W_{\sss HF}\| \!=\! \bm{O}( \frac{\epsilon }{k_s\sqrt{L_c}})  \!+\! \bm{O} ( \frac{\epsilon}{k_s}) + \bm{O} ( \frac{\sigma}{k_s})+ \bm{o}( \frac{\sigma^2}{k_s^2}).
\end{equation}
Taking $\epsilon$, $\sigma$ and $L_c$ as fixed constants, (\ref{eq:4terms}) can be further characterized as $\| \hat W_{\sss HF} -  W_{\sss HF}\|= \bm{O} ( \frac{1}{k_s}) + \bm{o}( \frac{1}{k_s^2})$. 
The proof is completed. 
\end{proof}

\subsection{Proof of Theorem \ref{th:decreasing-error}}\label{apdix:decreasing-error}
\begin{proof}
To begin with, suppose $R_c \ge \hat{R}_{c1}\ge  \hat{R}_{c2}$ and the observations are noise-free for ease notation. 
Denote by $\mathcal{V}_{\sss F_2}$ the robot set that is within the range $R_{h0}+\hat{R}_{c2}$, 
\textcolor{black}{and define $\mathcal{N}_{\sss H_0}^{in}=\{j:j\in\mathcal{V}_{\sss F}, i\in\mathcal{V}_{\sss H_0}, w_{ij}>0\}$. 
Then, the filtered observations of $\mathcal{V}_{\sss H_0}$ satisfy}
\begin{equation}\label{eq:nonoise-case}
Y_{\sss H_0} = {  W_{\sss{H_0 F_2}} } X_{\sss{F_2}}  + {  W_{\sss{H_0 F'_2}} } X_{\sss{F'_2}},
\end{equation}
where $\mathcal{V}_{\sss F'_2}=\{\mathcal{N}_{\sss H_0}^{in} \}\cap \{\mathcal{V}_{\sss F}\backslash\mathcal{V}_{\sss F_2}\}$. 
Note that ${  W_{\sss{H_0 F'_2}} } X_{\sss{F'_2}}$ is a ${|\mathcal{V}_{H_0}|}$-dimension vector. 
There exists at least a $|\mathcal{V}_{H_0}| \times | \mathcal{V}_{F_2} |$-dimension matrix $W_{\Delta_0 \sss{F_2}} $, such that 
\begin{equation}\label{eq:change-W}
W_{\Delta_0 \sss{F_2}}   Y_{\sss{F_2}} =  {  W_{\sss{H_0 F'_2}} } X_{\sss{F'_2}}. 
\end{equation} 
Then, substitute (\ref{eq:change-W}) into (\ref{eq:nonoise-case}) and it yields  
\begin{equation}\label{eq:nonoise-case2}
Y_{\sss H_0} = ( {  W_{\sss{H_0 F_2}} } + W_{\Delta_0 \sss{F_2}} )X_{\sss{F_2}}. 
\end{equation}
Taking the observation noises into account, the OLS estimator of $W_{\sss{H_0 F_0}}(\hat{R}_{c2})$ with observations $Y_{\sss H_0}$ and $X_{\sss{F_2}}$ is given by 
\begin{equation}
\hat W_{\sss{H_0 F_2}}(\hat{R}_{c2}) = Y_{\sss H_0} X_{\sss F_2}^\mathsf{T} ( X_{\sss F_2} X_{\sss F_2}^\mathsf{T})^{-1}. 
\end{equation}

Next, applying Theorem \ref{th:final-error} on $\hat W_{\sss{H_0 F_2}}(\hat{R}_c) $, it yields that  
\begin{equation}\label{eq:local-error}
\Pr\{ \| \mathop {\lim } \limits_{ k_s \to \infty } {\hat W_{\sss{H_0 F_2}}(\hat{R}_c) }- {  W_{\sss{H_0 F_2}} } \| =\| W_{\Delta_0 \sss{F_2}}\|\}=1.  
\end{equation}
It is straightforward to infer that, if some in-neighbors of $\mathcal{V}_{\sss H_0}$ are in $\mathcal{V}_{\sss F'_2}$, then $\| W_{\Delta_0 \sss{F_2}}\|> 0$. 
Therefore, for the local topology estimator $\hat W_{\sss{H_0 F}}(\hat{R}_{c2})$, the asymptotic inference bias is given by 
\begin{equation}
f_w(\hat{R}_{c2})=  \left \| \left [ W_{\Delta_0 \sss{F_2}} , W_{\sss{H_0 F'_2}}  \right] \right \|  .  
\end{equation}
For $\hat{R}_{c1}\ge \hat{R}_{c2}$, the cardinal numbers of $\mathcal{V}_{\sss F'_1}$ and $\mathcal{V}_{\sss F'_2}$ satisfy $|\mathcal{V}_{\sss F'_1}| \le |\mathcal{V}_{\sss F'_2}|$, which indicates that $\hat W_{\sss{H_0 F}}(\hat{R}_{c1})$ is less biased from the real $W_{\sss{H_0 F}}$. 
The monotone decreasing property of $f_w(\hat{R}_c)$ is proved. 

Specifically, if $\hat{R}_{c2}\ge R_{c}$, the set $\mathcal{V}_{\sss F'_2}=\{\mathcal{N}_{\sss H_0}^{in} \}\cap \{\mathcal{V}_{\sss F}\backslash\mathcal{V}_{\sss F_2}\}=\emptyset$. 
Finally, it follows that $|\mathcal{V}_{\sss F'_2}|=0$ and $f_w(\hat{R}_{c2})=0$, i.e., $\hat W_{\sss{H_0 F}}(\hat{R}_c) $ is asymptotically unbiased. 
The proof is completed. 
\end{proof}



\begin{IEEEbiographynophoto}{Yushan Li}
(S'19) received the B.E. degree in School of Artificial Intelligence and Automation from Huazhong University of Science and Technology, Wuhan, China, in 2018. 
He is currently working toward the Ph.D. degree with the Dept. of Automation, Shanghai Jiaotong University, Shanghai, China. 
He is a member of Intelligent of Wireless Networking and Cooperative Control group. 
His research interests include robotics, security of cyber-physical system, and distributed computation and optimization in multi-agent networks. 
\end{IEEEbiographynophoto}

\begin{IEEEbiographynophoto}{Jianping He} 
(SM'19) is currently an associate professor in the Department of Automation at Shanghai Jiao Tong University. He received the Ph.D. degree in control science and engineering from Zhejiang University, Hangzhou, China, in 2013, and had been a research fellow in the Department of Electrical and Computer Engineering at University of Victoria, Canada, from Dec. 2013 to Mar. 2017. His research interests mainly include the distributed learning, control and optimization, security and privacy in network systems.

Dr. He serves as an Associate Editor for IEEE Open Journal of Vehicular Technology and KSII Trans. Internet and Information Systems. He was also a Guest Editor of IEEE TAC, International Journal of Robust and Nonlinear Control, etc. He was the winner of Outstanding Thesis Award, Chinese Association of Automation, 2015. He received the best paper award from IEEE WCSP'17, the best conference paper award from IEEE PESGM'17, and was a finalist for the best student paper award from IEEE ICCA'17.
\end{IEEEbiographynophoto}
\vspace{-20pt}


\begin{IEEEbiographynophoto}{Cai Lin}
(F'20) received the M.A.Sc. and Ph.D. degrees (awarded Outstanding Achievement in Graduate Studies) in electrical and computer engineering from the University of Waterloo, Waterloo, Canada, in 2002 and 2005, respectively. Since 2005, she has been with the Department of Electrical and 
Computer Engineering at the University of Victoria, where she is currently a Professor. 
She is an NSERC E.W.R. Steacie Memorial Fellow. Her research interests span several areas in communications and networking, with a focus on network protocol and architecture design supporting emerging multimedia traffic and the Internet of Things. 

She was a recipient of the NSERC Discovery Accelerator Supplement (DAS) Grants in 2010 and 2015, respectively, and the best paper awards of IEEE ICC 2008 and IEEE WCNC 2011. She has co-founded and chaired the IEEE Victoria Section Vehicular Technology and Communications Joint Societies Chapter. 
She has been elected to serve the IEEE Vehicular Technology Society Board of Governors, 2019 - 2021. She has served as an Area Editor for IEEE TRANSACTIONS ON VEHICULAR TECHNOLOGY, a member of the Steering Committee of the IEEE TRANSACTIONS ON BIG DATA (TBD) and IEEE TRANSACTIONS ON CLOUD COMPUTING (TCC), an Associate Editor of the IEEE INTERNET OF THINGS JOURNAL, IEEE TRANSACTIONS ON WIRELESS COMMUNICATIONS, IEEE TRANSACTIONS ON VEHICULAR TECHNOLOGY, IEEE TRANSACTIONS ON COMMUNICATIONS, EURASIP Journal on Wireless Communications and Networking, International Journal of Sensor Networks, and Journal of Communications and Networks (JCN), and as the Distinguished Lecturer of the IEEE VTS Society. She has served as a TPC co-chair for IEEE VTC2020-Fall, and a TPC symposium co-chair for IEEE Globecom’10 and Globecom’13. She is a Registered Professional Engineer in British Columbia, Canada.
\end{IEEEbiographynophoto}
\vspace{-20pt}

\begin{IEEEbiographynophoto}{Xinping Guan}
(F'18) received the B.S. degree in Mathematics from Harbin Normal University, Harbin, China, in 1986, and the Ph.D. degree in Control Science and Engineering from Harbin Institute of Technology, Harbin, China, in 1999. He is currently a Chair Professor with Shanghai Jiao Tong University, Shanghai, China, where he is the Dean of School of Electronic, Information and Electrical Engineering, and the Director of the Key Laboratory of Systems Control and Information Processing, Ministry of Education of China. 
Before that, he was the Professor and Dean of Electrical Engineering, Yanshan University, Qinhuangdao, China. 

Dr. Guan's current research interests include industrial cyber-physical systems, wireless networking and applications in smart factory, and underwater networks. He has authored and/or coauthored 5 research monographs, more than 270 papers in IEEE Transactions and other peer-reviewed journals, and numerous conference papers. 
As a Principal Investigator, he has finished/been working on many national key projects. He is the leader of the prestigious Innovative Research Team of the National Natural Science Foundation of China (NSFC). 
Dr. Guan received the First Prize of Natural Science Award from the Ministry of Education of China in both 2006 and 2016, and the Second Prize of the National Natural Science Award of China in both 2008 and 2018. 
He was a recipient of IEEE Transactions on Fuzzy Systems Outstanding Paper Award in 2008. He is a National Outstanding Youth honored by NSF of China, Changjiang Scholar by the Ministry of Education of China and State-level Scholar of New Century Bai Qianwan Talent Program of China.
\end{IEEEbiographynophoto}

\end{document}